\documentclass[sigconf]{acmart}

\usepackage{amsmath}
\usepackage{amssymb,natbib,graphicx,subcaption,booktabs,url}
\usepackage[ruled,lined,boxed,linesnumbered]{algorithm2e}
\def \setC{\mathcal{C}}
\def \setD{\mathcal{D}}

\def \setG{\mathcal{G}}
\def \setN{\mathcal{N}}
\def \setS{\mathcal{S}}
\def \setV{\mathcal{V}}

\def \setE{\mathcal{E}}

\def\BibTeX{{\rm B\kern-.05em{\sc i\kern-.025em b}\kern-.08emT\kern-.1667em\lower.7ex\hbox{E}\kern-.125emX}}
    
\copyrightyear{2020} 
\acmYear{2020} 
\setcopyright{acmcopyright}\acmConference[KDD '20]{Proceedings of the 26th ACM SIGKDD Conference on Knowledge Discovery and Data Mining}{August 23--27, 2020}{Virtual Event, CA, USA}
\acmBooktitle{Proceedings of the 26th ACM SIGKDD Conference on Knowledge Discovery and Data Mining (KDD '20), August 23--27, 2020, Virtual Event, CA, USA}
\acmPrice{15.00}
\acmDOI{10.1145/3394486.3403100}
\acmISBN{978-1-4503-7998-4/20/08}

\begin{document}
\fancyhead{}

\title{Mining Large Quasi-cliques with Quality Guarantees from Vertex Neighborhoods}

\author{Aritra Konar}
\affiliation{%
  \institution{University of Virginia}
  \city{Charlottesville}
  \state{Virginia}
  \country{USA}
}
\email{aritra@virginia.edu}

\author{Nicholas D. Sidiropoulos}
\affiliation{%
  \institution{University of Virginia}
  \city{Charlottesville}
  \state{Virginia}
  \country{USA}
}
\email{nikos@virginia.edu}

%
% The abstract is a short summary of the work to be presented in the article.
\begin{abstract}
 Mining dense subgraphs is an important primitive across a spectrum of graph-mining tasks. In this work, we formally establish that two recurring characteristics of real-world graphs, namely heavy-tailed degree distributions and large clustering coefficients, imply the existence of substantially large vertex neighborhoods with high edge-density. This observation suggests a very simple approach for extracting large quasi-cliques: simply scan the vertex neighborhoods, compute the clustering coefficient of each vertex, and output the best such subgraph. The implementation of such a method requires counting the triangles in a graph, which is a well-studied problem in graph mining. When empirically tested across a number of real-world graphs, this approach reveals a surprise: vertex neighborhoods include maximal cliques of non-trivial sizes, and the density of the best neighborhood often compares favorably to subgraphs produced by dedicated algorithms for maximizing subgraph density.
 For graphs with small clustering coefficients, we demonstrate that small vertex neighborhoods can be refined using a local-search method to ``grow'' larger cliques and near-cliques. Our results indicate that contrary to worst-case theoretical results, mining cliques and quasi-cliques of non-trivial sizes from real-world graphs is often not a difficult problem, and provides motivation for further work geared towards a better explanation of these empirical successes.  
\end{abstract}


% Keywords.

%
% This command processes the author and affiliation and title information and builds
% the first part of the formatted document.
\maketitle

\section{Introduction}
\noindent \textbf{Motivation and Overview:}
The task of extracting dense subgraphs from a given graph constitutes a key primitive in graph mining, with applications ranging from graph compression \cite{buehrer2008scalable}, to discovering protein complexes in protein-protein interaction networks \cite{bader2003automated,prvzulj2004functional}, to identifying spam farms in Web graphs \cite{gibson2005discovering,hooi2016fraudar}, and event detection in network streams \cite{angel2012dense,cadena2016dense}. 

Depending on the particular metric employed for quantifying subgraph density, various formulations have been proposed for extracting different classes of dense subgraphs. The archetypal dense subgraph is a clique, i.e., a subgraph where every pair of vertices share an edge. A clique is said to be maximal if it isn't included within a larger clique, and the largest such clique is the maximum clique of a graph. The set of all maximal cliques in a graph can be listed using the classic Bron-Kerbosch algorithm \cite{bron1973algorithm}, albeit at  exponential worst-case complexity. Meanwhile, the problem of extracting the maximum clique is NP--hard \cite{garey2002computers}-- even for power-law graphs \cite{ferrante2008hardness}. 

Consequently, a different line of work has focused on developing less stringent, polynomial-time formulations for mining dense subgraphs. The seminal work of Goldberg \cite{goldberg1984finding} established that the problem of finding the subgraph with maximum average degree (widely known as the \textsc{DensestSubgraph} problem) can be solved via a sequence of maximum-flow problems. Follow-up work by Charikar \cite{charikar2000greedy} showed that a simple greedy vertex-peeling algorithm, that runs in linear-time, provides a $1/2$ approximation for  the problem and is near-optimal in practice. However, it was pointed out in \cite{tsourakakis2013denser} that adopting such a metric in practice can potentially yield the entire graph as the densest subgraph. As a result, Tsourakakis \cite{tsourakakis2015k} introduced the more general problem of finding the subgraph which maximizes the average number of induced $k$-cliques (known as the $k$-\textsc{CliqueDensestSubgraph} problem), and provided exact flow-based algorithms and greedy approximation algorithms for the task. It was also shown that this approach yields smaller, denser subgraphs compared to \textsc{DensestSubgraph}.

Another line of work utilizes a different relaxation of the notion of a clique, known as quasi-cliques, to find dense subgraphs. Formally, a $\alpha$-quasi-clique is a subgraph with edges greater than a fixed fraction $\alpha \in (0,1)$ of the edges in a clique of the same size. Recently, Tsourakakis \emph{et al.} introduced the \textsc{OptimalQuasiClique} (OQC) formulation in \cite{tsourakakis2013denser} for mining quasi-cliques possessing a large number of edges with respect to a random null model. The OQC problem is not known to be NP--hard; however, to the best of our knowledge, it does not admit an exact solution in polynomial-time either. Tsourakakis \emph{et al.} \cite{tsourakakis2013denser} proposed a simple greedy vertex-peeling algorithm (\textsc{GreedyOQC}) and a local-search method (\textsc{LocalSearchOQC}) for extracting approximate solutions for the problem, and demonstrated that they can work well in practice. Later, Cadena \emph{et al.} \cite{cadena2016dense} applied semidefinite relaxation (SDR) \cite{luo2010semidefinite} to the problem, and provided sufficient conditions under which SDR can guarantee a high-quality approximate solution. However, the high complexity incurred in solving the semidefinite program is a limitation of the approach.

\noindent \textbf{Approach and Contributions:} In this paper, we study the general problem of mining dense subgraphs from undirected graphs. In contrast to the prevailing approaches outlined above, we advocate a very simple method which can be summarized as follows: visit every vertex in the graph, compute the edge-density of the subgraph induced by its one-hop neighbors, and output the ``best'' (in a certain sense). This simply entails computing the local clustering coefficient \cite{newman2018networks} of every vertex, which can be accomplished by enumerating all triangles in the graph -- a task for which there exist several efficient algorithms \cite{latapy2008main,MACEcode2015}.  

While the approach may seem $\emph{apriori}$ naive  (it only considers one-hop neighborhoods), we provide theoretical justification for it by establishing the following result: if a graph possesses a large global clustering coefficient \cite{watts1998collective} and a heavy-tailed degree distribution \cite{barabasi1999emergence} (two recurring traits of real-world networks \cite{faloutsos1999power,watts1998collective}), then it includes large and dense vertex neighborhoods. Our work is motivated by the result of \cite{gleich2012vertex}, which established that the aforementioned properties of real-world networks imply that neighborhood subgraphs form communities with low conductance scores. However, to the best of our knowledge, the question of whether these properties also imply that vertex neighborhoods themselves constitute large and dense subgraphs (in the sense of being quasi-cliques) has not been studied prior to our present work. More specifically, our result differs from that  of \cite{gleich2012vertex} in the following aspects. 
\begin{itemize}
    \item The authors of \cite{gleich2012vertex} use a probabilistic existence argument to show that high global clustering coefficients and power-law degree distributions imply that there exists a vertex neighborhood with \emph{low conductance}. While we utilize the same probabilistic argument and the same twin graph characteristics, our result formally shows the existence of neighborhoods of non-trivial sizes possessing \emph{high edge-density}, which is a very different metric than conductance, and  necessitates a different line of analysis compared to that used in \cite{gleich2012vertex}.
    \item In \cite{gleich2012vertex}, it is also shown that the aforementioned properties of a graph imply the existence of a $k$-core\footnote{A $k$-core is the maximal subgraph of a graph where every vertex is connected to at least $k$ other vertices.}, which is a particular type of dense subgraph. Here, we restrict our attention to vertex neighborhoods, and adopt the edge-density of a subgraph as our notion of density. In general, these two notions of density are not directly comparable. 
    %as it is more general \footnote{Every $k$-core defines an $\alpha$-quasi-clique (with a particular $\alpha)$, but not all such $\alpha$-quasi-cliques are $k$-cores.}, thereby encompassing a wider variety of subgraphs.
    Moreover, the result of \cite{gleich2012vertex} relies on an argument that requires the graph to grow asymptotically in size. In contrast, we provide a non-asymptotic analysis to establish our result, albeit at the expense of making an explicit assumption on the power-law exponent of the degree distribution. 
\end{itemize}
It has further been shown \cite{gupta2014decompositions} that irrespective of the degree distribution, graphs with high global clustering coefficients admit a decomposition as a union of vertex disjoint subgraphs, each of which is guaranteed to possess a certain minimum edge and triangle density. We point out that where neighborhoods are concerned, high edge and triangle density are necessary, but not sufficient to guarantee the presence of dense neighborhoods of non-trivial sizes. As a counter-example, consider a graph which is a union of disjoint $4$-cliques. In this case, the global clustering coefficient is the maximum possible value $1$, and each vertex neighborhood is simply a triangle, which also attains maximum edge and triangle density. To rule out such unfavorable cases, we employ the power-law degree assumption, which is commonly observed in many real-world networks.    

In order to test our hypothesis regarding the existence of such large neighborhood subgraphs with high edge-density, and to gauge the empirical efficacy of our approach, we carried out a series of experiments on $15$ different publicly available datasets, with the \textsc{GreedyOQC} algorithm of \cite{tsourakakis2013denser} and the sophisticated maximum flow-based algorithm of \cite{mitzenmacher2015scalable} for computing the triangle-densest subgraph \cite{tsourakakis2015k} used as benchmarks. We point out that these baselines are not neighborhood based, and constitute dedicated algorithms for dense subgraph discovery. Our main empirical findings can be summarized as follows:
\begin{itemize}
    \item For graphs which obey our sufficient conditions, we discovered that neighborhoods can surprisingly form \emph{maximal cliques} and quasi-cliques of non-trivial sizes. Furthermore, the quality of these neighborhood subgraphs is comparable, or even better compared to the baselines. While these results validate the essence of our theoretical argument, they also reveal the conservative nature of our analysis, as we obtain better results in practice. 
    \item For graphs with low global clustering coefficients, neighborhoods with high local clustering coefficients can be small in size. However, we demonstrate that they can serve as good seed sets for a local-search algorithm proposed in \cite{tsourakakis2013denser}. We provide empirical justification for our choice by demonstrating that it is consistently better in terms of size and edge-density compared to subgraphs obtained via other simple seeding strategies such as the core decomposition and selecting neighborhoods with high average degree. Refining our neighborhoods via this algorithm allows us to obtain cliques and near-cliques of even better quality compared to the baselines.
\end{itemize}
Finally, we note that, while the scope of our algorithmic contributions is limited, our main purpose is to highlight the fact that substantially large dense neighborhoods exist in real-world graphs. On the theoretical side, we provide practical sufficient conditions on the graph characteristics (in terms of power-law degree distributions and clustering coefficients) to quantify the existence of such large, dense neighborhoods. On the practical side, via extensive experiments, we verify that such neighborhoods are of comparable, or even better quality, compared to a range of baselines, and when refined using a local search algorithm they yield state-of-the-art results.      
Our findings suggest that contrary to worst-case complexity results \cite{garey2002computers,ferrante2008hardness,eppstein2010listing}, it is possible to extract large cliques and near-cliques from real-world graphs using a {\em very simple approach} -- and this is quite remarkable. %In particular, under appropriate conditions, graphs can harbor dense vertex neighborhoods. These neighborhoods can be identified via triangle enumeration, and in turn can be further refined using a simple local search algorithm to harvest a rich variety of dense subgraphs at low complexity.

\section{Preliminaries}

 Given a simple, unweighted, undirected graph $\setG := (\setV,\setE)$ on $n$ vertices, the \emph{neighborhood} of a vertex $v \in \setV$ is the subset of vertices $\setN_v \subseteq \setV$ that share an edge with $v$. This can be expressed as
 \begin{equation}
 \setN_v := \{ u \in \setV: (u,v) \in \setE\}, \forall \; v \in \setV.
 \end{equation}
 The degree of vertex $v \in \setV$ is $d_v:=|\setN_v|$. 
 A \emph{wedge} is a path of length $2$ formed by an unordered pair of edges $\{(s,v),(v,t)\}$ that share a common vertex $v$. A wedge is said to be \emph{closed} if its end points $(s,t)$ are connected by an edge. Let $w_v:=\binom{d_v}{2}$ denote the number of wedges centered at vertex $v$ and $w_v^{(c)}$ denote the corresponding number of closed wedges.
 The \emph{local clustering coefficient} of $v$ is then the fraction of wedges centered at $v$ that are closed, i.e.,
 \begin{equation}
 C_v := \frac{w_v^{(c)}}{w_v}, \forall \; v \in \setV.
 \end{equation}
 Let $w:=\sum_{v \in \setV}w_v$ be the total number of wedges in $\setG$. The \emph{global clustering coefficient} of $\setG$ is the overall fraction of wedges in $\setG$ that are closed, i.e.,
 \begin{equation}
 C_g := \frac{1}{w}\sum_{v \in \setV}w_v^{(c)}.
 \end{equation}
Define a probability mass function $p$ on the vertices of $\setG$ that assigns each vertex $v \in \setV$ a probability equal to the fraction of overall wedges centered at $v$, i.e.,
 \begin{equation}
 p_v := \frac{w_v}{w}, \forall \; v \in \setV.
 \end{equation}
 It is known \cite[Claim 4.2]{gleich2012vertex} that the above twin definitions of clustering coefficients obey the following relation with respect to (w.r.t.) the distribution $p$. 
 \begin{equation}\label{eq:clusters}
 \mathbb{E}_p[C_v] = C_g
 \end{equation}  
 Given a subset of vertices $\setS \subseteq \setV$, define ${\setE}({\setS})$ as the subset of $\setE$ containing edges only between the vertices in $\setS$. For the subgraph $\setG_{\setS} := (\setS,\setE(\setS))$ induced by $\setS$, let $e(\setS):=|\setE(\setS)|$ denote the number of edges in $\setG_{\setS}$. The density of a subgraph is measured via its edge-density
 \begin{equation}
 \delta(\setS):= \frac{e(\setS)}{\binom{|\setS|}{2}},
 \end{equation}  
 which quantifies how closely $\setG_S$ resembles a clique on $|\setS|$ vertices in terms of edges, i.e., $ \delta(\setS) = 1$ when $\setS$ is a clique. Given a parameter $\alpha \in (0,1)$, a subgraph $\setG_S$ is said to be a $\alpha-$quasi-clique if $\delta(\setS) \geq \alpha$, i.e., if the number of its edges is at least as large as a fixed fraction $\alpha$ of the edges in a clique on $|\setS|$ vertices.
  
\section{Vertex neighborhoods as dense subgraphs}
In this section, we analyze whether vertex neighborhoods themselves can be potential candidates for dense subgraphs in real-world graphs. Our starting point is the following simple observation which states that the edge-density of a vertex neighborhood equals its local clustering coefficient. 
\begin{lemma}
For all $\setS = \setN_v$, $ \delta(\setS) = C_v$.
\end{lemma}
\begin{proof}
Observe that every edge	in $\setN_v$ induces a closed wedge centered at $v$, which implies that $e(\setN_v) = w_v^{(c)}$. Furthermore, as $d_v = |\setN_v|$, we have $\binom{|\setN_v|}{2} = w_v$.
\end{proof}
\noindent If we treat $C_v$ as a random variable with distribution $p$, an immediate consequence of the above lemma and \eqref{eq:clusters} is the following equation
\begin{equation}\label{eq:avg_density}
\mathbb{E}_{p}[\delta(\setN_v)] = C_g,
\end{equation}
which implies that for graphs with large global clustering coefficients, the edge-density of a vertex neighborhood is also large on average. 
 If a vertex $v \in \setV$ is sampled with probability $p_v$, we can establish the following bounds on the probability of $\setN_v$ being an $\alpha-$quasi-clique. 
\begin{lemma}
For all $\alpha > C_g$,
\begin{equation}\label{eq:markov}
\textrm{Pr}\{\delta(\setN_v) \geq \alpha\} \leq \frac{C_g}{\alpha}, \forall \; v \in \setV.
\end{equation}
Meanwhile, for $\alpha < C_g$,
\begin{equation}\label{eq:lb}
\textrm{Pr}\{\delta(\setN_v) \geq \alpha\} \geq \frac{C_g-\alpha}{1-\alpha}, \forall \; v \in \setV.
\end{equation} 
\end{lemma}
\begin{proof}
The upper bound \eqref{eq:markov} follows as a simple consequence of Markov's inequality. To establish the lower bound \eqref{eq:lb}, we use the following result extracted from \cite[Theorem 4.6]{gleich2012vertex} 
\begin{equation}\label{eq:eventA}
\textrm{Pr}\{C_v \leq \alpha\} \leq \frac{1-C_g}{1-\alpha}.
\end{equation}
Combining the above inequality with Lemma 3.1 yields the desired claim.
\end{proof}
\noindent Clearly, the lower bound \eqref{eq:lb} is more informative compared to the upper bound \eqref{eq:markov}, as Markov's inequality typically yields a loose bound on the tail probability. Note that for large $C_g$, the lower bound \eqref{eq:lb} can yield a non-trivial result. This can be observed from Figure \ref{fig:fig2}, which illustrates the bounds as a function of $\alpha$ for $C_g = 0.7$. For example, when $\alpha = 2/3$, observe that the probability of a vertex neighborhood $\setN_v$ being a $2/3-$quasi-clique is at least $10\%$. It is also evident that the bounds diverge as $\alpha$ approaches the mean $C_g$. It is only in the extreme case of $C_g = 1$, that the bounds coincide to yield $\textrm{Pr}\{\delta(\setN_v) \geq \alpha\} = 1$. This result can be explained by the fact that for $C_g = 1$, the graph $\setG$ is a union of disjoint cliques. Consequently, any vertex neighborhood is also a clique (being the subgraph of a clique), which is always a quasi-clique for every choice of $\alpha$. 

\begin{figure}[t!]
	\centering
	\includegraphics[width = 0.3\textwidth]{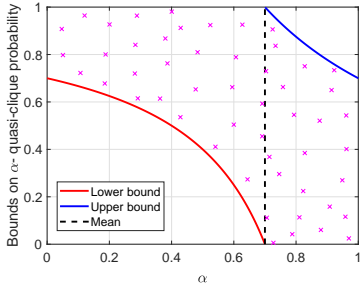}
	\caption{\footnotesize Illustration of the upper bound \eqref{eq:markov} and lower bound \eqref{eq:lb} on the probability of a vertex neighborhood being a $\alpha-$quasi-clique for $C_g = 0.7$. The purple crosses mark the feasible region.} 
	\label{fig:fig2} 	
\end{figure}

 Additional insight regarding the behavior of the distribution about the mean $C_g$ can be obtained by analyzing the variance of $\delta(\setN_v)$. To this end, we will require the following result.
\begin{lemma}
	$\mathbb{V}_p[\delta(\setN_v)] \leq C_g(1-C_g)$
\end{lemma}
\begin{proof}
	Note that the second-order moment of the random variable $C_v$ can be bounded as 
	\begin{subequations}
		\begin{alignat}{7}
		\mathbb{E}_{p}[C_v^2] &= \sum_{v \in \setV} p_vC_v^2 
		\leq \sum_{v \in \setV} p_vC_v  = C_g,
		\end{alignat}
	\end{subequations}
	where the inequality stems from the fact that $C_v \in [0,1], \forall \; v \in \setV$. Combining the result with \eqref{eq:clusters} and Lemma 1, we obtain
	\begin{equation}
	\begin{aligned}
	\mathbb{V}_p[\delta(\setN_v)] &= \mathbb{E}_{p}[C_v^2] - (\mathbb{E}_{p}[C_v])^2 \\
	&\leq C_g(1-C_g),
	\end{aligned}
	\end{equation}
	which establishes the desired claim.
\end{proof}

\noindent The result implies that for low $C_g$, the variance is small, and thus the values of $\delta(\setN_v)$ are likely to be ``close'' to the mean $C_g$. In other words, it is unlikely that many neighborhoods exhibit high edge-density. Conversely, as the obtained bound is symmetric about $C_g = 1/2$, for large $C_g$, the vertex neighborhoods with edge-density close to $C_g$ are likely candidates for being dense subgraphs.

While the aforementioned results suggest that graphs with high global clustering coefficients harbor potentially many dense vertex neighborhoods, as pointed out in the introduction, high edge-density alone is a necessary, but not sufficient condition for the existence of large, dense vertex neighborhoods.

Thus far, our analysis has only been reliant on the clustering coefficient of a graph. We now attempt to incorporate another salient characteristic of real-world graphs into our analysis: heavily-skewed degree distributions. It is well known that the degree distribution of many graphs can be well approximated by a power-law \cite{watts1998collective}. Let $(d_{\min}, d_{\max})$ denote the smallest degree $>1$ and the largest degree of a graph respectively, and let $\setD:=\{d_{\min}, \cdots,d_{\max}\}$  denote the set of unique degrees in $\setG$. For a given degree $d \in \setD$, let $n_d$ denote the number of times a vertex $v \in \setV$ takes value $d$. In order to facilitate analysis, we make the following simplifying assumptions:
\begin{enumerate}
	\item[(C1)] The power law exponent of the degree distribution of $\setG$ is $\gamma = 2$, which is fairly reasonable as $\gamma$ typically takes values in the range $[1.75,3]$ for real world networks \footnote{This choice is made for convenience and brevity of exposition; we can handle other values of $\gamma > 2$ as well, but the derivations are more cumbersome, see Remark 1.}. This enables us to express 
	\begin{equation}\label{eq:power_law}
	n_d = c n d^{-2}, \forall\; d \in \setD,
	\end{equation}
	where $c \in \mathbb{R}$ denotes the normalization constant of the degree distribution.
	\item[(C2)] The set $\setD$ does not contain any ``missing'' degrees, i.e, there exists a vertex of degree $d$ for every possible choice of $d$ satisfying $ d_{\min} \leq d \leq d_{\max}$. 
\end{enumerate}
  
Our objective is now to combine both aspects (skewed degree distributions and high clustering coefficients) to formally establish the existence of vertex neighborhoods of non-trivial sizes with high edge density. In order to do so, we take recourse to the probabilistic method \cite{alon2016probabilistic}, a classical and powerful technique in combinatorics for certifying the existence of combinatorial objects possessing certain properties within a probability space. We proceed by first  defining the following pair of ``bad'' events:
\begin{enumerate}
	\item[(A)] a vertex sampled with probability $p_v$ has a neighborhood with ``low'' edge-density,
	\item[(B)] a vertex sampled with probability $p_v$ has a ``small'' degree, i.e., a neighborhood of small size.  
\end{enumerate}	
If we can establish that the probability of either event occurring is strictly less than 1, then it implies the existence of a vertex neighborhood which simultaneously possesses high edge-density and non-trivial size. The exact notions of ``low'' edge-density and ``small'' neighborhood size will be quantified next.   

Note that \eqref{eq:eventA} already provides an upper bound on the probability of event A occurring. We now seek to establish an upper bound on the probability of event B. For a given parameter $\beta \in \biggl(\frac{d_{\min}}{d_{\max}},1\biggr)$, define $\bar{d}:= \beta d_{\max}$ and let $\setS_{\bar{d}}$ denote the set of all vertices having degree greater than $1$ but lesser than equal to $\bar{d}$, i.e.,
\begin{equation}
\setS_{\bar{d}} := \{v \in \setV: d_{\min} \leq d_v \leq \bar{d}\}.
\end{equation}
We also define a set ${\bar{\setD}} \subseteq \setD$ to be the subset of all unique degrees of $\setG$ not exceeding $\bar{d}$, i.e.,
\begin{equation}
\bar{\setD} := \{d \in \setD: d_{\min} \leq d \leq \bar{d}\}.
\end{equation}
Armed with these definitions, we can derive the following upper bound on the probability of sampling a vertex with a degree smaller than a fraction $\beta$ of the largest degree $d_{\max}$.
\begin{lemma}
	$\textrm{Pr}\{v \in \setS_{\bar{d}}\} <  \frac{\beta d_{\max}  - \log \beta}
	{d_{\max} - \log\biggl(\frac{d_{\min}}{d_{\min} - 1}\biggr)} $
\end{lemma}
\begin{proof}
The probability of event B can be expressed as 
\begin{equation}\label{eq:eventB}
\textrm{Pr}\{v \in \setS_{\bar{d}}\} = \sum_{v \in \setS_{\bar{d}}}p_v 
= \sum_{v \in \setS_{\bar{d}}} \frac{w_v}{w}
= \frac{\sum_{v \in \setS_{\bar{d}}}w_v}{\sum_{v \in \setV} w_v}.
\end{equation}
Exploiting the twin facts that $w_v = \binom{d_v}{2}, \forall\; v \in \setV$ and that the degree distribution of $\setG$ obeys a power law of the form \eqref{eq:power_law}, we obtain the following expressions for the numerator and denominator of \eqref{eq:eventB}
\begin{equation}
\begin{aligned}
\sum_{v \in \setS_{\bar{d}}}w_v &= \frac{c n}{2}\sum_{d \in \bar{\setD}}d(d-1)d^{-2} 
= \frac{c n}{2}\sum_{d \in \bar{\setD}} \biggl(1-\frac{1}{d}\biggr),\\
 \sum_{v \in \setV}w_v &= \frac{c n}{2}\sum_{d \in \setD} d(d-1)d^{-2} 
= \frac{c n}{2}\sum_{d \in \setD} \biggl(1-\frac{1}{d}\biggr).
\end{aligned}
\end{equation} 
This allows us to further simplify \eqref{eq:eventB} to
\begin{equation}\label{eq:eventB2}
\textrm{Pr}\{v \in \setS_{\bar{d}}\} 
= \frac{|\bar{\setD}| - \sum_{d_{\min}}^{\bar{d}}1/d } {|{\setD}| - \sum_{d_{\min}}^{d_{\max}}1/d}.
\end{equation}
In order to derive an upper bound on \eqref{eq:eventB2}, we exploit the following general fact regarding partial harmonic sums (see \cite[Appendix A, p. 1154]{cormen2009introduction})
\begin{equation}
\int_{l}^{u+1}\frac{dx}{x} \leq \sum_{n=l}^{u}\frac{1}{n} \leq \int_{u-1}^{l} \frac{dx}{x}
\end{equation}
where $(l,u)$ are integers that obey $1<l<u$ and denote the lower and upper limits of the sum respectively. On computing the integrals, we obtain the approximation bounds
\begin{equation}
\log\biggl(\frac{u+1}{l}\biggr) \leq \sum_{n=l}^{u}\frac{1}{n} \leq \log\biggl(\frac{u}{l-1}\biggr)
\end{equation}
Applying the lower bound to the partial harmonic sum appearing in the numerator and the upper bound to the one in the denominator of \eqref{eq:eventB2}, we obtain 
\begin{equation}\label{eq:chain1}
\textrm{Pr}\{v \in \setS_{\bar{d}}\} \leq \frac{|\bar{\setD}|-\log\biggl(\frac{\bar{d}+1}{d_{\min}}\biggr)}
{|\setD|-\log\biggl(\frac{d_{\max}}{d_{\min}-1}\biggr)}
\end{equation}
The upper bound obtained above can be further bounded by applying the following chain of (strict) inequalities
\begin{equation}\label{eq:chain2}
\begin{aligned}
\frac{|\bar{\setD}|-\log\biggl(\frac{\bar{d}+1}{d_{\min}}\biggr)}
{|\setD|-\log\biggl(\frac{d_{\max}}{d_{\min}-1}\biggr)}
 &< \frac{|\bar{\setD}|-\log\biggl(\frac{\bar{d}}{d_{\min}}\biggr)}
{|\setD|-\log\biggl(\frac{d_{\max}}{d_{\min}-1}\biggr)}\\
&= \frac{|\bar{\setD}| - \log \beta - \log\biggl(\frac{d_{\max}}{d_{\min}}\biggr)}
{|\setD| -\log\biggl(\frac{d_{\max}}{d_{\min}-1}\biggr)}\\
&= \frac{|\bar{\setD}| - \log \beta - \log\biggl(\frac{d_{\max}}{d_{\min}}\biggr)}
{|\setD| - \log\biggl(\frac{d_{\min}}{d_{\min}-1}\biggr) - \log\biggl(\frac{d_{\max}}{d_{\min}}\biggr)}\\
&< \frac{|\bar{\setD}| - \log \beta}
{|\setD| - \log\biggl(\frac{d_{\min}}{d_{\min}-1}\biggr)}.
\end{aligned}
\end{equation}
Upon defining $\Delta := \frac{d_{\min}}{d_{\min}-1}$, and using the fact that 
$$|\bar{\setD}| = \bar{d}-d_{\min}+1 = \beta d_{\max} - d_{\min}+1,$$
$$ |{\setD}| = d_{\max} - d_{\min}+1, $$
it simply remains to apply the chain of inequalities derived in \eqref{eq:chain1} and \eqref{eq:chain2} to finally obtain the claimed upper bound on the probability of event B
\begin{equation}
\begin{aligned}
\textrm{Pr}\{v \in \setS_{\bar{d}}\} &< \frac{\beta d_{\max} - d_{\min}+1 - \log \beta}
{d_{\max} - d_{\min}+1 - \log\Delta}\\
&< \frac{\beta d_{\max}  - \log \beta}
{d_{\max} - \log\Delta}
\end{aligned}
\end{equation}

\end{proof}
\noindent \textbf{Remark 1:} Our assumption regarding the value of the power-law exponent can be relaxed to any value $\gamma > 2$ to obtain a result of a similar flavor, at the expense of a more cumbersome analysis. Owing to space constraints, we only sketch the requisite modifications. The key difference for $\gamma >2$ is that the functions being summed in the numerator and denominator of \eqref{eq:eventB2} are now $d^{2-\gamma}$ and $d^{1-\gamma}$, which are non-increasing in $d$ for $\gamma >2$. For such functions, the integral approximation trick borrowed from \cite[Appendix A, p. 1154]{cormen2009introduction} still applies, and consequently, can again be used to derive an upper bound on \eqref{eq:eventB2}. The exact form of the bound is dependent on the specific value of $\gamma$ used, as this determines the form that the integrals ultimately take.

Back to our present case of $\gamma =2$, define the quantities $\eta: = \frac{\beta d_{\max}  - \log \beta}{d_{\max} - \log\Delta}$, and $\beta_{\max}$ to be the largest value of $\beta$ that satisfies $\eta < C_g$.
With Lemma 3.4 in hand, we can establish the following theorem.
\begin{theorem}
Under assumptions (C1) and (C2), there exists a vertex neighborhood of size $|\setN_v| \geq \beta d_{\max}$ and edge-density $\delta(\setN_v) \geq \frac{C_g - \eta}{1-\eta}$, for every choice of $ \beta \in \biggl(\frac{d_{\min}}{d_{\max}},\beta_{\max}\biggr)$.
\end{theorem}
\begin{proof}
 Since $|\setN_v| = d_v, \forall \; v \in \setV$, from Lemma 7 we obtain
\begin{equation}
\textrm{Pr}\{v \in \setS_{\bar{d}}\}  = \textrm{Pr}\{ d_{\min} \leq |\setN_v| \leq \beta d_{\max} \}< \eta.
\end{equation}
Meanwhile, on setting $\alpha := \frac{C_g - \eta}{1-\eta}$ in \eqref{eq:eventA}, we obtain
\begin{equation}
\textrm{Pr}\{\delta(\setN_v) \leq \alpha\} \leq 1- \eta.
\end{equation}
A simple application of the union bound then reveals that the probability of either of the above events occurring is strictly less than $1$, thus implying that the complement ``good'' event occurs with positive probability. Hence, there exists a vertex neighborhood of size $|\setN_v| \geq \beta d_{\max}$ which is at least a $\frac{C_g - \eta}{1-\eta}$ quasi-clique. 
\end{proof}
When $d_{\max}$ is large, then $\eta \approx \beta$, and thus the quasi-clique value (roughly) varies like $\frac{C_g - \beta}{1-\beta}$. In this case, $\beta_{\max} \approx C_g$, with the result that the allowable range of $\beta$ is the interval $ \biggl(\frac{d_{\min}}{d_{\max}},C_g\biggr)$. A limitation of our result is that it does not allow us obtain results for $\beta > C_g$. However, for large $C_g$, we obtain a non-trivial lower bound on the size of $\setN_v$ and its edge-density.
As an illustration of our lower bound for a real graph, please refer to Figure \ref{fig:bound} in the supplement.

Additionally, we point out an interesting fact about vertex neighborhoods: if a neighborhood $\setN_v$ forms a clique on $k$-vertices, then $\setN_v \cup \{v\}$ is a clique on $(k+1)$-vertices, which we designate as an \emph{ego-clique}. The following result asserts that such ego-cliques must be maximal.
\begin{theorem}
Let $\setN_v$ be a clique on $k$-vertices and $\setC_{k+1}(v):=\setN_v \cup \{v\}$ be an ego-clique on $(k+1)$-vertices. Every such ego-clique is maximal.
\end{theorem}
\begin{proof}
Assume the contrary, i.e., that there exists a clique $\setC_{\ell} \subset \setV$ on $\ell$-vertices such that $\ell > k+1 $ and $\setC_{k+1}(v) \subset \setC_{\ell}$. Then, there exists a vertex $u \in \setC_{\ell} \setminus \setC_{k+1}(v)$ which is one-hop away from $v$, since $v \in \setC_{\ell}$. This implies that $u \in \setN(v) \subset \setC_{k+1}(v)$, which is a contradiction. 
\end{proof}

\section{Experimental Evaluation}
In this section, we devise a series of experiments on a variety of datasets that aims to address the following questions: (a) Do dense vertex neighborhoods of non-trivial sizes exist in real-world graphs?  (b) How does the approach fare in comparison to dedicated algorithms for dense subgraph discovery?

\subsection{Datasets}
 The list of datasets used and a summary of their statistics are presented in Table \ref{tab:stats}. If the original graph is directed, a symmetrization step is first performed. Unless specified, the datasets were obtained from \cite{snapnets}, and can be classified as follows:
\begin{enumerate}
    \item[(A)] \textbf{Co-authorship graphs:} The vertices denote scientists, and the edges represent collaborations between co-authors of a scientific publication. The datasets include co-authorship graphs constructed from arXiv submissions in three different scientific disciplines (\textsc{arxiv-HepPh}, \textsc{arxiv-AstroPh} and \textsc{arxiv-CondMat}), as well as larger graphs comprising the largest connected component of the arXiv and DBLP co-authorship graphs (\textsc{arXiv} \cite{esfandiar2010fast} and \textsc{dblp} respectively).   
    \item[(B)] \textbf{Social networks:} The vertices are people, and the edges indicate ``friend'' relationships. The datasets include two different snapshots of the Facebook friendship graph (\textsc{Facebook-A} and \textsc{Facebook-B} \cite{viswanath2009evolution}), friendship networks obtained from a blogging website (\textsc{blogCatalog3} \cite{tang2009relational}), and a location-based social networking website (\textsc{loc-Gowalla}).
    %, and a photo-sharing website (\textsc{Flickr} \cite{esfandiar2010fast}). 
    \item[(C)] \textbf{Web graphs:} Vertices are web pages, while the edges denote symmetrized hyperlinks (\textsc{web-Stanford} and \textsc{web-Google}). 
    \item[(D)] \textbf{Miscellaneous:} An assortment of graphs drawn from different domains: a protein-protein interaction network (\textsc{ppi-Human}), an email communications network (\textsc{email-Enron}), a router graph (\textsc{router-Caida} \cite{davis2011university}), and an item-item co-purchase network (\textsc{Amazon}).
\end{enumerate}
\subsection{Assessing the Quality of Neighborhood Subgraphs}

Given a dataset, we first compute the edge-density of all vertex neighborhoods. This requires calculating the local clustering coefficient of every vertex, which can be accomplished by  triangle counting - a task that incurs a worst-case complexity of $O(m^{3/2})$ for a graph with $m$ edges. For our purposes, we employed the MAximal Clique Enumerator (MACE) algorithm (the C code of which is publicly available at \cite{MACEcode2015}) to obtain triangle counts. % We observed that MACE is very efficient in practice, taking $30s$ to list all $1.3 \times 10^7$ triangles in the largest graph \textsc{web-Google} on a Linux laptop.

Next, for every unique degree in the graph, we compute the highest neighborhood edge-density score over all vertices of that degree and display this information on a plot versus the log of the unique degrees. We designate such a plot as the \emph{neighborhood density profile} (NDP) of a graph, which is shown for six datasets in Figure \ref{fig:ndp}. The NDP plots in the first column represent graphs with high global clustering coefficients, which serve as good test beds for our working hypothesis that vertex neighborhoods are dense subgraphs.
Meanwhile, the graphs in the second column possess very low global clustering coefficients, and illustrate the outcome when our sufficient conditions for high neighborhood edge-density are not met.
In each NDP plot, we mark the largest degree $d_{\textrm{max}}$ by a vertical magenta line, the size of the largest clique discovered by the \textsc{greedyOQC} algorithm (for comparison) using a vertical red line, while the global clustering coefficient is highlighted using a black horizontal line. 
\begin{table}
  \caption{\footnotesize Summary of graph statistics: the number of vertices ($n$), the number of edges ($m$), the largest degree ($d_{\textrm{max}}$), the global clustering coefficient ($C_g$), and the mean local clustering coefficient $\bar{C}$.}
  \label{tab:stats}
  \footnotesize
  \begin{tabular}{cccccl}
    \toprule
    Graph & $n$ & $m$ & $d_{\textrm{max}}$ & $C_g$ & $\bar{C}$\\
    \midrule
    \textsc{arXiv-HepPh} & 12,008 & 112K & 491 & 0.659 & 0.612 \\
    \textsc{arXiv-AstroPh} & 18,772 & 198K & 504 & 0.318 & 0.677\\
    \textsc{arXiv-CondMat} & 23,133 & 93,497 & 279 & 0.264 & 0.633\\
    \textsc{arXiv} & 86,376 & 517K & 1,253 & 0.560 & 0.678\\
    \textsc{dblp} & 317K & 1.05M & 343 & 0.306 & 0.632\\
    \midrule
    \textsc{Facebook-A} & 4,039 & 88,234 & 1,045 & 0.519 & 0.605\\
    \textsc{blogCatalog3} & 10,312 & 333K & 3,992 & 0.091 & 0.463\\
    \textsc{Facebook-B} & 63,731 & 817K & 1,098 & 0.148 & 0.221\\
    \textsc{loc-Gowalla} & 196K & 950K & 14,730 & 0.023 & 0.237 \\
   % \textsc{Flickr} & 513K & 3.19M & 4,369 & 0.159 & 0.168\\
    \midrule
    \textsc{web-Stanford} & 281K & 2.31M & 38,625 & 0.008 & 0.598\\
    \textsc{web-Google} & 875K & 5.10M & 6,332 & 0.055 & 0.514\\
    \midrule
    \textsc{ppi-Human} & 21,557 & 342K & 2,130 & 0.119 & 0.207 \\
    \textsc{email-Enron} & 36,692 & 183K & 1,383 & 0.085 & 0.497\\
    \textsc{router-Caida} & 192K & 609K & 1,071 & 0.061 & 0.157\\
    \textsc{Amazon} & 334K & 923K & 549 & 0.205 & 0.397\\
  \bottomrule
\end{tabular}
\end{table}
A feature common to all NDP plots is that the neighborhood edge-density decreases with increase in  degree, which follows from the fact that the local clustering coefficient of a vertex is inversely proportional to the square of its degree. However, when the global clustering coefficient of the graph is large, from the NDP plots in the first column, it is evident that vertex neighborhoods themselves constitute large (relative to the largest degree $d_{\textrm{max}}$), dense subgraphs. In fact, it can be observed that several neighborhoods $\setN(v)$ attain an edge-density equal to 1, i.e., they form a clique. Recalling the result of Theorem 3.6, it then follows that for the \textsc{arXiv-HepPh}, and \textsc{dblp} datasets, inspecting vertex neighborhoods alone surprisingly reveals maximal cliques of non-trivial sizes.
Furthermore, for these datasets, the size of the largest such ego-clique matches the result obtained using the \textsc{greedyOQC} algorithm. On the other hand, for the $\textsc{facebook-A}$ dataset, the size of the largest ego-clique is roughly $6-$times smaller than that obtained by \textsc{greedyOQC}. However, it can be seen that there do exist vertex neighborhoods of size comparable to that of the clique discovered by \textsc{greedyOQC}, which are $0.9$-quasi-cliques, and thus, are also substantially dense. Taken together, the NDP plots in the first column of Figure \ref{fig:ndp} provide empirical validation of our hypothesis that graphs with power-law degree distributions and high global clustering coefficients harbor large, dense neighborhood subgraphs. 
\begin{figure}[t!]
	\includegraphics[width = 0.23\textwidth]{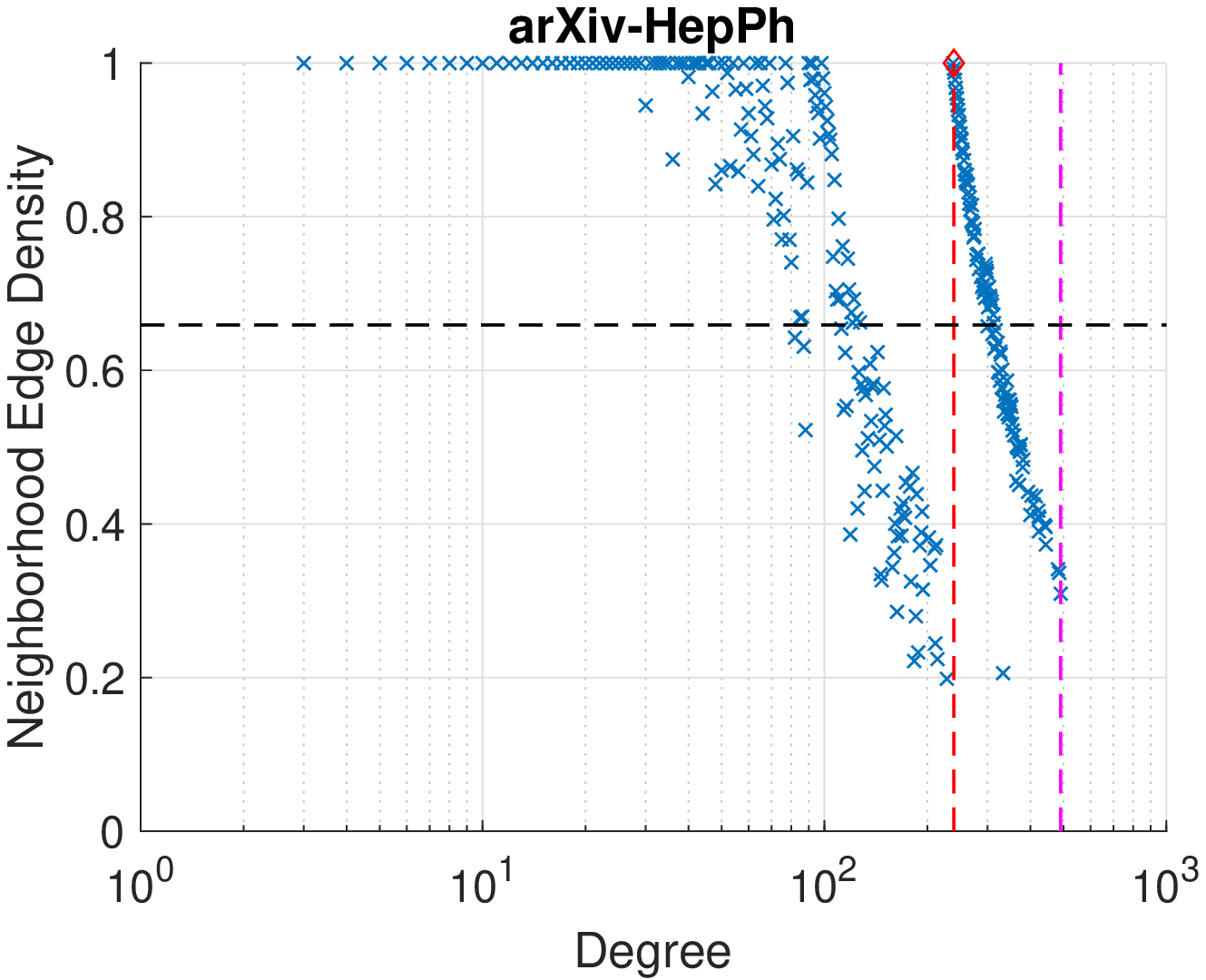}
	\includegraphics[width = 0.23 \textwidth]{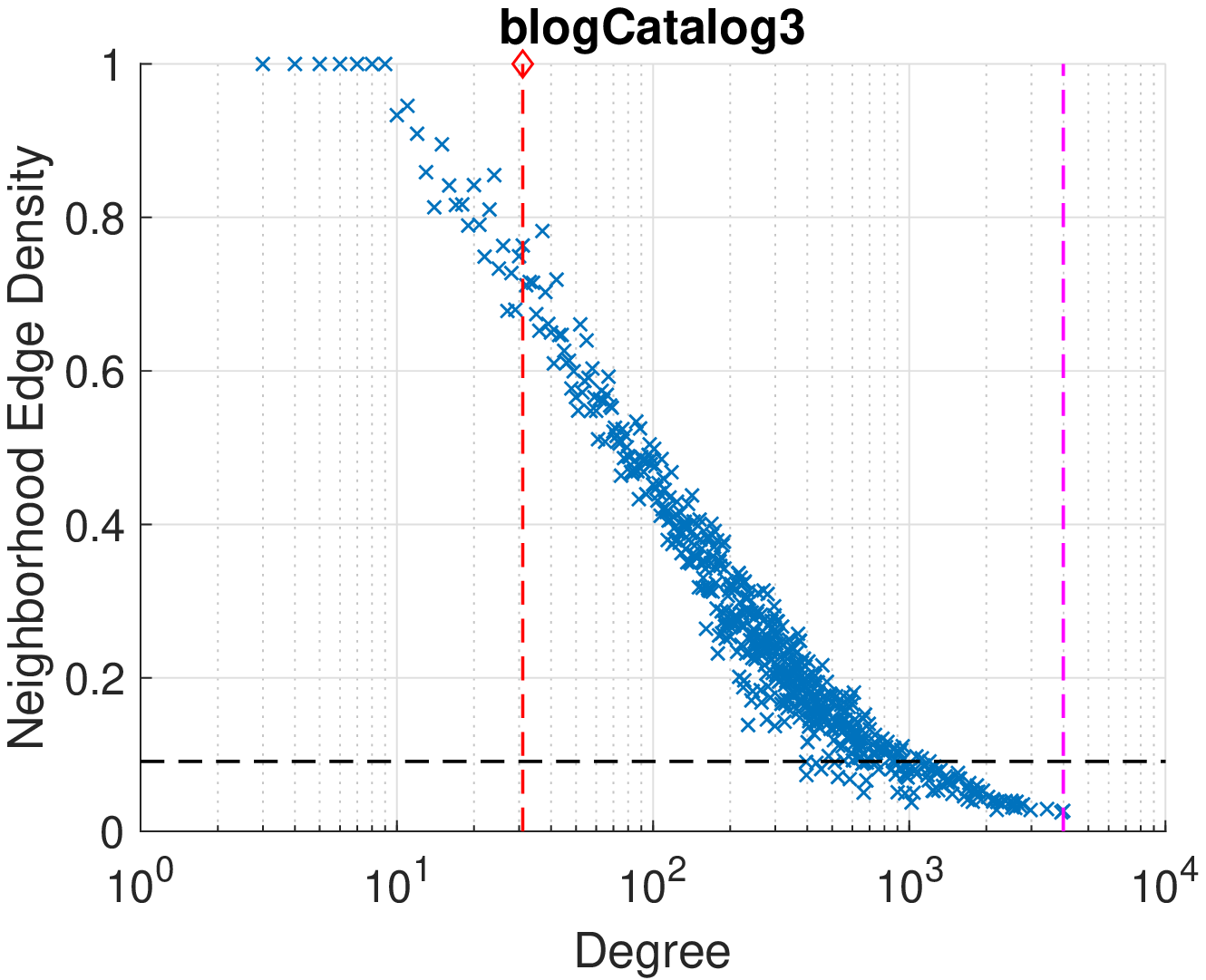}
	\includegraphics[width = 0.23 \textwidth]{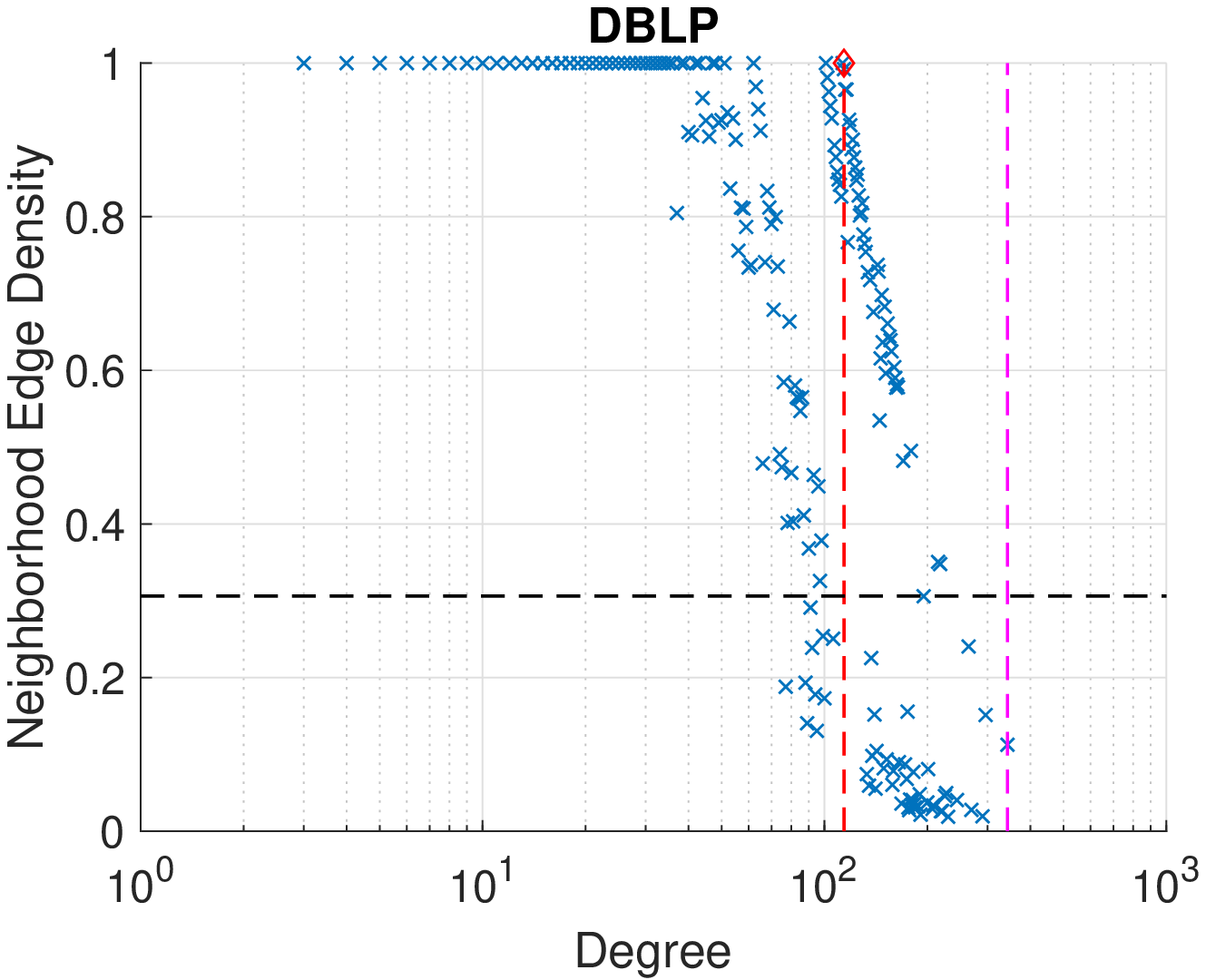}
	\includegraphics[width = 0.23\textwidth]{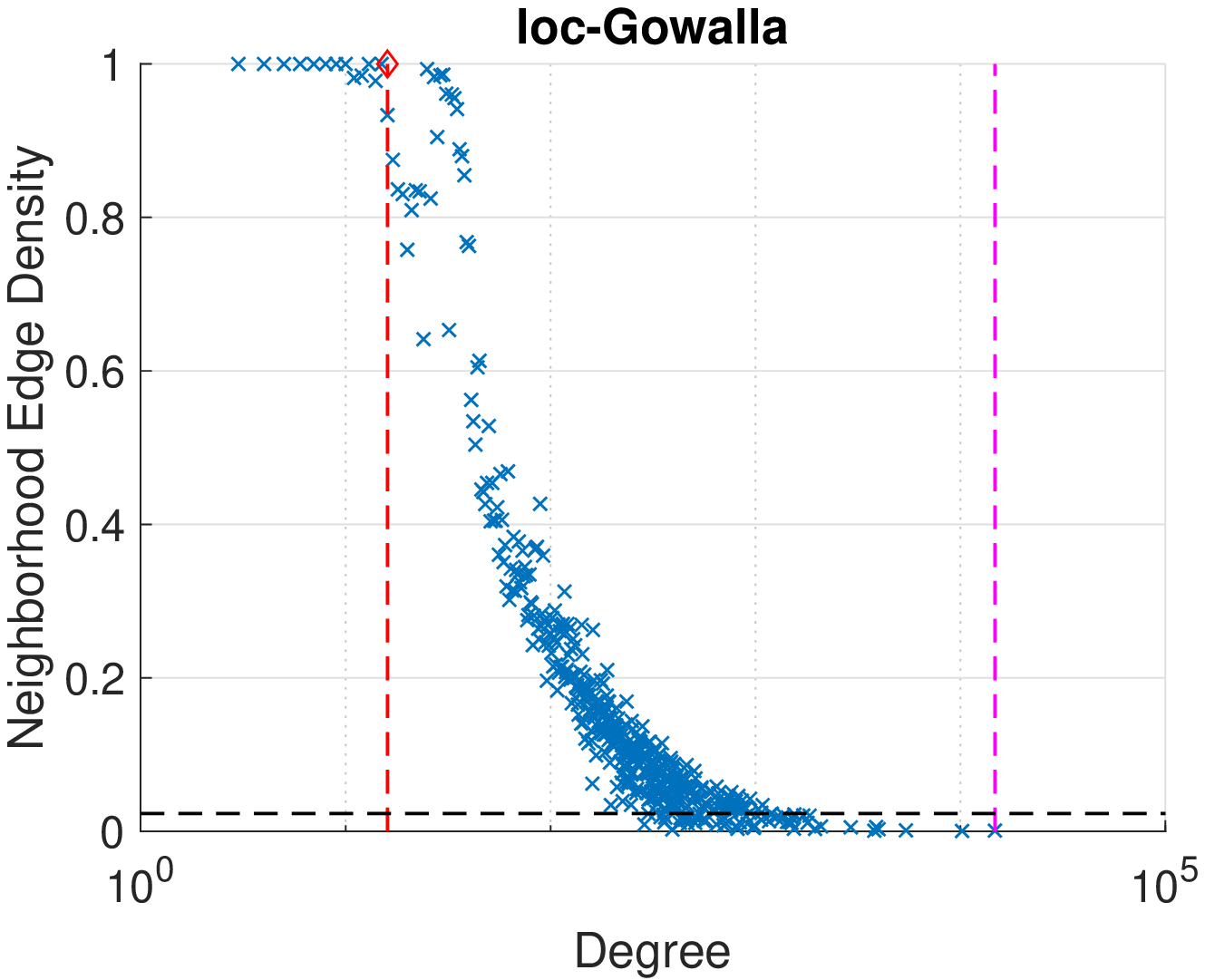}
	\includegraphics[width = 0.23 \textwidth]{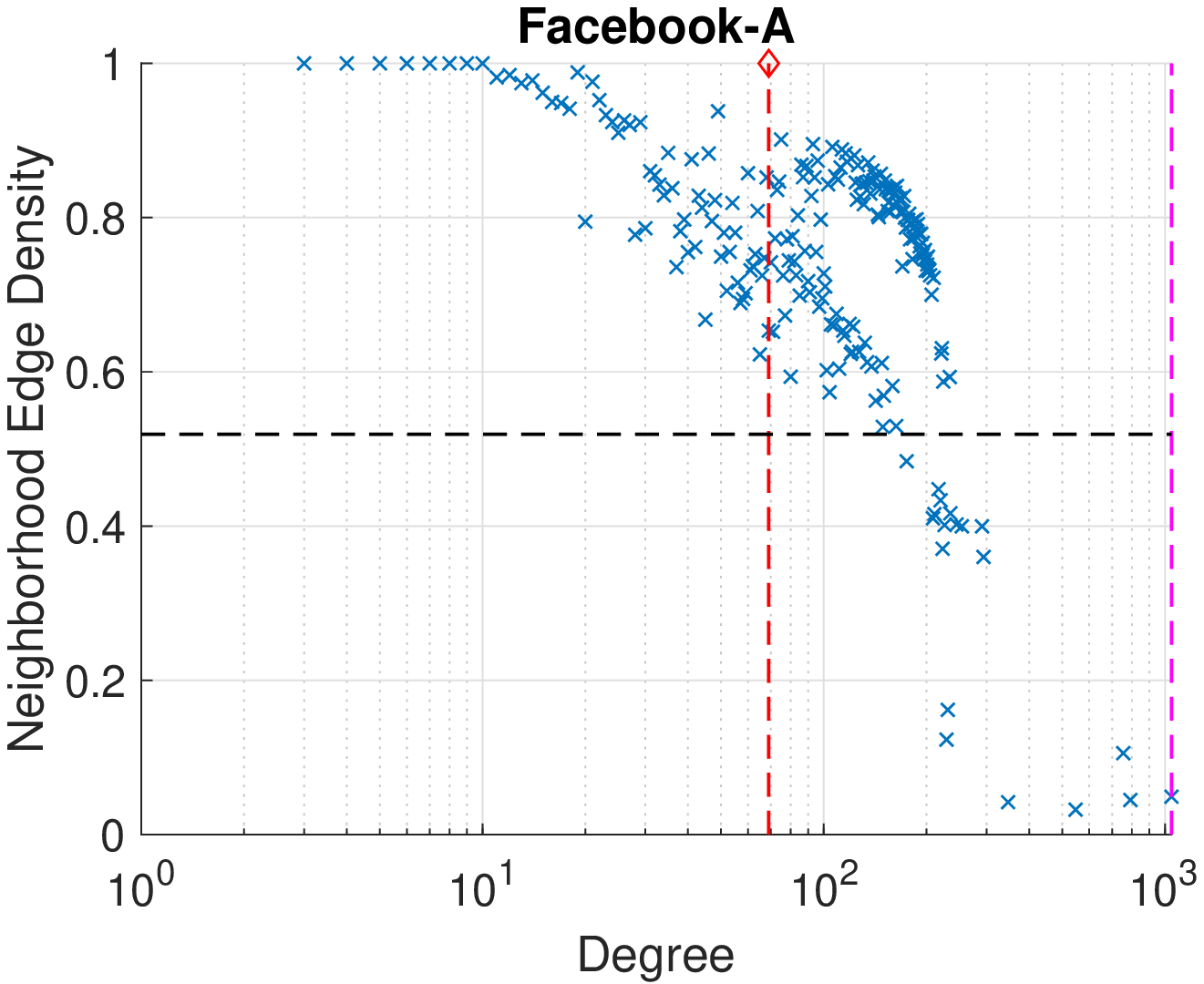}
	\includegraphics[width = 0.23 \textwidth]{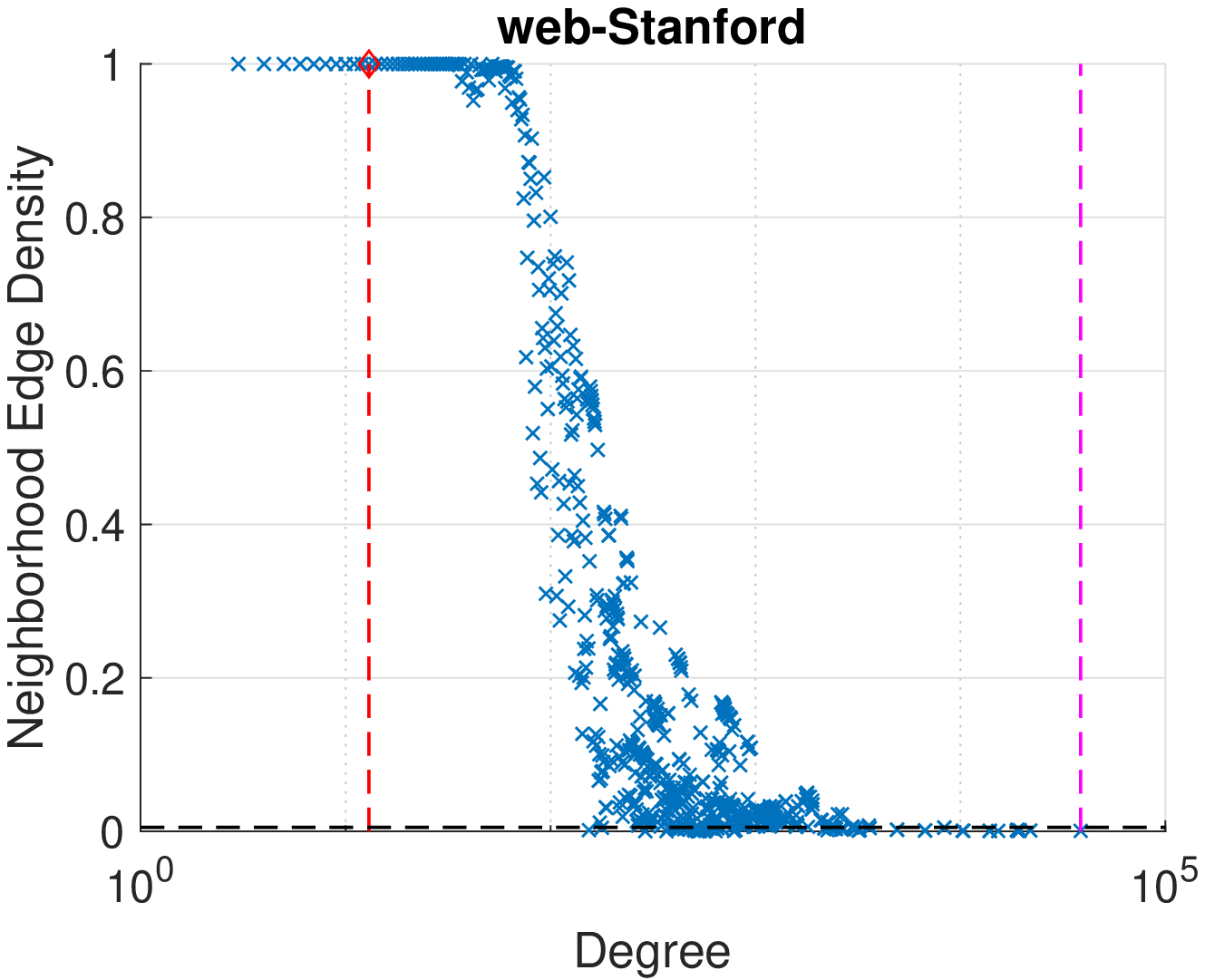}
    \caption{\footnotesize The Neighborhood Density Profile of six real-world graphs. Each plot depicts the maximum of the edge-density of vertex neighborhoods of a given degree versus the log of the degrees. Horizontal black line -- global clustering coefficient ($C_g$), {\color{red} red} vertical line -- densest subgraph returned by the \textsc{greedyOQC} algorithm, and the {\color{magenta} magenta} vertical line -- largest degree $d_\textrm{max}$. The graphs in the first column figures have high $C_g$ values, while the ones in the second column have small $C_g$ values.}
    \label{fig:ndp}
\end{figure}

We now turn our attention to the second column, where $C_g$ is small. In this case, note that the size of the densest neighborhood subgraph is small with respect to the largest degree $d_{\textrm{max}}$. In particular, for the \textsc{blogCatalog3} graph, the NDP reveals that the neighborhood edge-density decays quickly with the degree. This represents the worst-case scenario, where the vertex neighborhoods themselves are not appealing candidates for being dense subgraphs of non-trivial sizes. On the other hand, for the graphs \textsc{loc-Gowalla}, and \textsc{web-Stanford}, 
%the edge-density exhibits a steep drop-off beyond a certain degree. In these cases, 
there are a few dense vertex neighborhoods which form small (relative to $d_\textrm{max}$) subgraphs of non-trivial sizes, and represent atypical or ``anomalous'' regions of the graph. Note that in terms of quality, on the \textsc{loc-Gowalla} graph, the densest vertex neighborhood is near-optimal in terms of size and edge-density compared to the solution returned by \textsc{greedyOQC}, while on the  \textsc{web-Stanford} graph, the largest ego-clique is 4 times larger in size compared to the clique computed by \textsc{greedyOQC}. 
%Meanwhile on the \textsc{ppi-Human} graph, the size of the largest ego-clique is $81$, which, while certainly non-trivial, is smaller than the clique on $130$  vertices discovered by \textsc{greedyOQC}. Again, in this case, there are a few neighborhood subgraphs of similar size which are near-cliques ($0.9$-quasi-cliques).   

In summary, the NDP of a graph is very informative in assessing the edge-density of neighborhood subgraphs. It reveals the presence of large, dense neighborhood subgraphs in real-world graphs with power-law degree distributions and high global clustering coefficients, thereby confirming the essence of the result provided by Theorem 3.5. Moreover, it illustrates that graphs exhibiting the aforementioned traits often feature the surprising attribute that neighborhoods themselves constitute maximal cliques of non-trivial sizes, with the size of the largest clique being the same as that determined by \textsc{greedyOQC}, which is a non-neighborhood based method. On the other hand, it also showcases that when $C_g$ is small, then neighborhood subgraphs may still form small, dense subgraphs of non-trivial sizes. That being said, there also exist unfavorable instances where neighborhood subgraphs are dense only on a very small scale. The following section explores ways of using such neighborhood subgraphs as seeds for a local-search method in order to grow dense subgraphs of larger sizes.

\begin{table}[b!]
	\caption{\footnotesize Quality of subgraphs obtained via core decomposition and selecting neighborhoods based on average degree in terms of their size $|\setS|$ and edge-density $|\delta(\setS)|$.}
	\footnotesize
	\label{tab:results2}
	\begin{tabular}{cccccl}
		\toprule
		&  \multicolumn{2}{c}{Core decomposition} &  & \multicolumn{2}{c}{Avg. degree}\\
		\cmidrule{2-3} \cmidrule{5-6}
		Graph & $|\setS|$ & $\delta(\setS)$ & &  $|\setS|$ &  $\delta(\setS)$\\
		\midrule
		\textsc{arXiv-AstroPh} & 57 & 1 & & 81 & 0.75 \\
		\textsc{arXiv} & 146 & 0.49 & & 147 & 0.52 \\
		\textsc{blogCatalog3} & 447 & 0.4 & & 1550 & 0.08 \\
		\textsc{Facebook-B} & 699 & 0.12 & & 723 & 0.07 \\
		\textsc{loc-Gowalla} & 183 & 0.41 & & 162 & 0.27\\
		\textsc{web-Stanford} & 387 & 0.29 & & 694 & 0.17 \\
		\textsc{router-Caida} & 92 & 0.45 & &  91 & 0.31 \\		
		\textsc{Amazon} & 497 & 0.013 & & 47 & 0.20 \\
		\bottomrule
	\end{tabular}
\end{table}
\subsection{Growing Dense Subgraphs from Vertex Neighborhoods}

In this section, we describe how the \textsc{localSearchOQC} algorithm of \cite{tsourakakis2013denser} can be used to refine the quality of vertex neighborhoods. Given an initial seed set $\setS_0 \subseteq \setV$, the \textsc{localSearchOQC} algorithm aims to maximize the edge-surplus objective function
\begin{equation}\label{eq:OQC}
f_{\alpha}(\setS): = e(\setS) - \alpha\binom{|\setS|}{2}
\end{equation}
by searching for vertices, which when added or deleted from the current solution set, yields an improvement in the objective function. The procedure is continued until a locally optimal solution is found, (i.e., until addition or deletion of a single vertex from the solution set does not lead to an improvement in the objective), or a maximum number of iterations $T_{\max}$ are reached. While the algorithm has a low run-time complexity of $O(mT_{\max})$,  its performance is particularly sensitive to the choice of initialization $\setS_0 \subseteq \setV$ as the objective function $f_{\alpha}(\setS)$ is difficult to maximize (globally). In that regard, we provide compelling empirical evidence that selecting vertex neighborhoods (on the basis of their clustering coefficients) constitutes good seeds for \textsc{localSearchOQC}. We devised a pair of simple strategies for judiciously selecting seed sets via this metric - owing to space limitations, the full details are provided in the supplement (see strategies \textbf{(S1)} and \textbf{(S2)}). 

In order to provide empirical justification for our choice, we performed a comparison against a pair of low-complexity alternatives for obtaining seed sets. These are (i) computing the core decomposition of the graph \cite{seidman1983network}, and (ii) selecting vertex neighborhoods on the basis of their average degree. The first choice is motivated by the fact that under the same general assumptions made in Theorem 3.5, a result of a similar flavor has been established in \cite{gleich2012vertex} regarding the existence of a dense core, and that the core decomposition can be computed efficiently in linear-time \cite{batagelj2003m}. As the procedure generates a hierarchy of nested subgraphs, we used the final subgraph in the hierarchy (which is the smallest in size and the densest) as a candidate seed. The second choice was proposed in \cite{tsourakakis2013denser} to initialize \textsc{localSearchOQC}; i.e., the neighborhood with the highest average degree is selected as the seed. Note that computing the average degree of a vertex neighborhood incurs the same complexity as computing clustering coefficients. However, selecting neighborhoods via this metric presently lacks theoretical justification, in contrast to ours. 
The quality of the best seeds obtained by the alternatives is depicted in Table \ref{tab:results2} -- these results are representative of both the best and worst outcomes. Meanwhile, the quality of the best neighborhood obtained obtained by employing strategy \textbf{(S2)} is depicted in Table \ref{tab:results} (see columns under Quasi-cliques with heading ``NB"). It is evident that our neighborhood selection strategy consistently yields seeds that are of considerably higher quality compared to those obtained via the alternatives (in terms of both size and edge-density). We conclude that our mechanism of generating seeds is well suited for providing high-quality initializations for \textsc{localSearchOQC} on real-world data compared to the prevailing baselines. 
Following the suggestion of \cite{tsourakakis2013denser}, we set the maximum number of iterations $T_{\max} = 50$ in our experiments.Apart from the choice of the initial seed set $\setS_0$, the performance of the algorithm is also dependent on the choice of the parameter $\alpha \in (0,1]$. The recommendation of \cite{tsourakakis2013denser} is to set $\alpha = 1/3$. However, we observed that in practice, on many graphs, the performance of the algorithm with neighborhood seeding can be significantly improved by simply increasing $\alpha$ to much larger values. 
 For a more thorough discussion on selecting $\alpha$, please refer to the supplement.

%Following the suggestion of Tsourakakis \emph{et. al} \cite{tsourakakis2013denser}, we set $T_{\max} = 50$ in our experiments.

\subsection{Main Results and Discussion}

We compared our approach against two non-neighborhood based methods -- the \textsc{greedyOQC} algorithm of \cite{tsourakakis2013denser} and a sophisticated flow-based algorithm proposed in \cite{mitzenmacher2015scalable} for efficently computing the $k$-clique densest subgraph \cite{tsourakakis2015k}. For the former algorithm, which employs greedy vertex peeling to maximize the OQC function \eqref{eq:OQC} and runs in linear time $O(m+n)$, we used a value of $\alpha > 1/3$, as it substantially improves upon the performance reported in \cite{tsourakakis2013denser} (see supplement for an example). Meanwhile, for a given integer $k \geq 3$, the latter method requires a list of all $k$-cliques in the graph as input. For fair comparison, we used $k=3$, which reduces to listing triangles, that we already obtained using MACE for computing clustering coefficients. Note that for this choice of $k$, the algorithm aims to compute the triangle-densest subgraph (TDS). We used the C- based implementation developed by the authors of \cite{mitzenmacher2015scalable} that is publicly available at \cite{flowcode2015} to obtain our results. 

We summarize the outcomes of our experiments across all datasets in Table \ref{tab:results}, which displays the size of the largest clique obtained by each method on each dataset, along with the ``best'' quasi-clique (i.e., the densest subgraph that is not a clique). The algorithm of \cite{mitzenmacher2015scalable} does not have any parameter to tune, and hence, we simply display the  obtained results. For \textsc{greedyOQC}, we report the largest clique obtained by setting $\alpha=1$. Meanwhile, for \textsc{localSearchOQC}, the cliques were obtained using the neighborhood seed sets \textbf{(S1)} and $\alpha = 1$, while the quasi-cliques were recovered using the neighborhood seed sets \textbf{(S2)}. We compared the quasi-cliques returned by the different methods on the basis of their size, edge-density and triangle-density. For fair comparison, we report the quasi-cliques obtained by each method for $\alpha = 0.9$ -- if a method returned a clique for this choice of $\alpha$, we used the next smaller value of $\alpha \in \{0.7,0.75,0.8,0.85\}$ for which a quasi-clique is obtained. If no quasi-clique is returned by a method for any choice of $\alpha$, we leave a blank in its corresponding location
in the table. Our main findings can be summarized as follows:
\begin{enumerate}
    \item The best neighborhood (\emph{without} refinement) is, in general, of much higher quality compared to the TDS computed by \cite{mitzenmacher2015scalable}, which requires triangle-listing as a pre-processing step. Furthermore, there always exists a high quality neighborhood quasi-clique (with $\alpha \geq 0.92$ in all but one case) of substantial size - refinement via \textsc{localSearchOQC} mainly yields a similar sized subgraph with improved triangle-density. Overall, these results provide empirical validation of our hypothesis that real-world graphs contain high-quality dense neighborhood subgraphs of non-trivial sizes.
      
    \item The \textsc{greedyOQC} algorithm (with appropriate tuning) is well-suited for clique discovery in general. However, on $6/15$ datasets, the largest clique discovered by \textsc{greedyOQC} \emph{and} \textsc{localSearchOQC} is no better than the largest ego-clique. On the remaining datasets, while the largest ego-clique can be small relative to \textsc{greedyOQC}, by using neighborhoods as seeds for \textsc{localSearchOQC}, we can discover a clique of comparable, or even larger size.
    \item Regarding the performance of \textsc{localSearchOQC} and \textsc{greedyOQC}, while both methods recover quasi-cliques of high quality, the former algorithm has a tendency to produce ``denser'' quasi-cliques of higher triangle density compared to the latter method. %This again highlights the complementary styles of the two approaches.
    \item On $7/15$ datasets (particularly, on collaboration networks), we observed that  \textsc{greedyOQC} produces a clique, but not any dense quasi-cliques, with the algorithm becoming ``stuck'' at the same clique for all choices of $\alpha$. Such an undesirable  behavior was not observed for \textsc{localSearchOQC}.
\end{enumerate}
To conclude, our results indicate that selecting vertex neighborhoods based on their local clustering coefficient reveals dense subgraphs of substantial size, which can be competitive with or even better than those obtained by dedicated methods for dense subgraph discovery. We also demonstrated that such vertex neighborhoods are good seeds for \textsc{localSearchOQC}, being substantially better overall than seeds obtained via other simple alternatives such as the core decomposition or choosing neighborhoods with large average degree. Further refining neighborhoods with this simple algorithm allows us to consistently obtain both cliques and quasi-cliques of even higher quality compared to the baselines across a wide variety of heterogeneous datasets.

%and then further refining them via \textsc{localSearchOQC} (with appropariate tuning) allows us to explore and discover a wealth of dense subgraphs, ranging from maximal cliques to near-cliques of different sizes. While the \textsc{greedyOQC} algorithm is competitive, its main empirical strengths lie in discovering cliques (although, even in this regard it outperforms the previous approach by a substantial margin on only $2/16$ datasets), with the algorithm failing to capture the broader spectrum of dense quasi-cliques in nearly half of our datasets. Hence, overall, using vertex neighborhoods as seeds for  \textsc{localSearchOQC} constitutes a simple but potent framework for dense subgraph discovery in undirected graphs. 

\begin{table*}
  \caption{\footnotesize Single best clique and quasi-clique computed by each method. The second column displays the clique size (the larger the better), while the last 3 columns display the quality of quasi-cliques as measured by their size $|\setS|$, edge-density $\delta(\setS) = e(\setS)/\binom{|\setS|}{2}$ and triangle-density  $\tau(\setS) = t(\setS)/\binom{|\setS|}{3}$ (here, $t(\setS)$ is the number of triangles in subgraph $\setS$). NB - neighborhood, NB+LS - local search with neighborhood seeds, GRDY - greedyOQC, TDS - flow based algorithm of \cite{mitzenmacher2015scalable} for computing the triangle densest subgraph.}  
  \footnotesize
  \label{tab:results}
  \begin{tabular}{ccccccccccccccccccc}
    \toprule
    &  \multicolumn{3}{c}{Cliques} & & \multicolumn{14}{c}{Quasi-cliques} \\
    \cmidrule{2-4} \cmidrule{6-19}
    &   \multicolumn{3}{c}{$|\setS|$} & & \multicolumn{4}{c}{$|\setS|$} & & \multicolumn{4}{c}{$\delta(\setS)$} & & \multicolumn{4}{c}{$\tau(\setS)$} \\
    \cmidrule{2-4} \cmidrule{6-9} \cmidrule{11-14} \cmidrule{16-19}  
    Graph & NB & NB + LS & GRDY & & NB & NB + LS & GRDY & TDS & & NB & NB + LS & GRDY & TDS & & NB & NB + LS & GRDY & TDS \\
    \midrule
    \textsc{arXiv-HepPh} & 239 & 239 & 239 & & 246 & 247 & - & 239 & & 0.95 & 0.95 & - & 1 & & 0.92 & 0.91 & - & 1\\
    \textsc{arXiv-AstroPh} & 57 & 57 & 57 & & 48 & 45 & - & 76 & & 0.90 & 0.99 & - & 0.80 & & 0.83 & 0.97 & - & 0.59\\
    \textsc{arXiv-CondMat} & 23 & \textbf{26} & \textbf{26} & & 19 & 18 & - & 30 & & 0.86 & 0.96 & - & 0.93 & & 0.68 & 0.89 & - & 0.72\\
    \textsc{arXiv} & 74 & 74 & 74 & & 75 & 60 & - & 146 & & 0.95 & 0.98 & - & 0.49 & & 0.92 & 0.94 & - & 0.25\\
    \textsc{dblp} & 114 & 114 & 114 & & 105 & - & - & 114 & & 0.95 & - & - & 1 & & 0.92 & - & - & 1\\
    \midrule
    \textsc{Facebook-A} & 11 & 32 & \textbf{69} & & 50 & 53 & 118 & 195 & & 0.94 & 0.98 & 0.97 & 0.79 & & 0.85 & 0.94 & 0.92 & 0.54\\
    \textsc{blogCatalog3} & 10 & 29 & \textbf{31} & & 12 & 52 & 52 & 621 & & 0.95 & 0.96 & 0.96 & 0.31 & & 0.87 & 0.88 & 0.88 & 0.05\\
    \textsc{Facebook-B} & 12 & \textbf{25} & \textbf{25} & & 20 & 17 & 36 & 198 & & 0.95 & 0.98 & 0.96 & 0.36 & & 0.85 & 0.95 & 0.89 & 0.08\\
    \textsc{loc-Gowalla} & 15 & \textbf{28} & 16  & & 36 & 32 & 23 & 311 & & 0.94 & 0.99 & 0.95 & 0.27 & & 0.85 & 0.97 & 0.86 & 0.04\\
    %\textsc{Flickr} & 11 & 39 & \textbf{53} & & 22 & 19 & 100 & & 0.93 & 0.99 & 0.95 & & 0.80 & 0.99 & 0.86\\
    \midrule
    \textsc{web-Stanford} & 53 & \textbf{53} & 14 & & 71 & 68 & 16 & 684 & & 0.95 & 0.99 & 0.96 & 0.17 & & 0.89 & 0.97 & 0.88 & 0.02\\
    \textsc{web-Google} & 25 & 43 & \textbf{44} & & 54 & 48 & 48 & 66 & & 0.93 & 0.99 & 0.99 & 0.85 & & 0.84 & 0.98 & 0.98 & 0.64\\
    \midrule
    \textsc{ppi-Human} & 81 & \textbf{130} & \textbf{130} & & 81 & - & - & 361 & & 0.93 & - & - & 0.42 & & 0.89 & - & - & 0.14\\
    \textsc{email-Enron} & 10 & \textbf{16} & \textbf{16} & & 14 & 12 & 22 & 388 & & 0.93 & 0.98 & 0.96 & 0.19 & & 0.82 & 0.95 & 0.89 & 0.02\\
    \textsc{router-Caida} & 9 & \textbf{15} & 6 & & 12 & 15 & - & 75 & & 0.92 & 0.97 & - & 0.55 & & 0.94 & 0.99 & 0.95 & 0.20\\
    \textsc{Amazon} & \textbf{7} & \textbf{7} & 5 & & 7 & 8 & 7 & 50 & & 0.95 & 0.96 & 0.90 & 0.19 & & 0.86 & 0.90 & 0.72 & 0.02\\
  \bottomrule
\end{tabular}
\end{table*}
\section{Conclusions}

Our main aim in this paper was to draw attention to the fact that real-world graphs harbor dense vertex neighborhoods of non-trivial sizes, which are often of comparable or higher quality relative to those discovered by dedicated algorithms for maximizing subgraph density. We provided theoretical justification of this phenomenon, in terms of sufficient conditions (namely, a power-law degree distribution and a large global clustering coefficient) under which such a surprising result can be expected in a real-world graph. In practice, our conditions seem to be conservative. We also provided compelling empirical evidence that refining a judiciously chosen neighborhood via a simple local search algorithm delivers state-of-the-art performance at low complexity. This indicates that discovering large cliques and near-cliques is not always hard for real-world graphs, and provides motivation for future work that provides a more refined analysis of these empirical results. 

\section{Acknowledgements}
Supported by the National Science Foundation and the Army Research Office under Grants No. IIS-1908070 and ARO-W911NF1910407 respectively. The authors additionally acknowledge the assistance of Charalampos Tsourakakis and Paris Karakasis in executing \cite{flowcode2015}.
\bibliographystyle{ACM-Reference-Format}
\bibliography{AKNS_KDD2020.bib}

\newpage

\section*{Supplementary Material}

In order to facilitate reproducibility, this section contains a detailed description of the mechanisms used to generate the neighborhood seed sets for initializing \textsc{localSearchOQC}, guidelines for choosing the tuning parameter $\alpha$ in the OQC objective function \eqref{eq:OQC}  for both the \textsc{localSearchOQC} and \textsc{greedyOQC} algorithms, and additional experiments showcasing how the choice of these parameters influences the obtained results.  Additionally, we provide an example to illustrate the quality of the lower bound on the neighborhood quasi-clique value derived in Theorem 3.5 on a real-world graph.

We begin by discussing the choice of $\alpha$ for \textsc{localSearchOQC}. While the recommendation of \cite{tsourakakis2013denser} is to set $\alpha = 1/3$, the algorithm performs much better in practice with a larger value. 
Such a beneficial effect can be partially explained via  the following intuitive argument: consider the case where $\bar{C} > 1/3$ for a given graph $\setG$. Note that the term $\alpha\binom{|\setS|}{2}$ in $f_{\alpha}(\setS)$ can be interpreted as the expected number of edges in a subgraph $\setG_\setS$ of a random Erdos-Renyi graph with edge-density $\alpha$. This random graph model serves as a null model which is used to compare and contrast the number of edges of a subgraph $\setG_\setS$ in the given graph $\setG$. We now point out that $\alpha$ can also be equivalently viewed as the expected local clustering coefficient of a random Erdos-Renyi graph. This observation suggests that given a graph, we can set the value of $\alpha$ to be equal to the average clustering coefficient $\bar{C}$ of $\setG$, as the random Erdos-Renyi graph model will exhibit the same clustering coefficient as $\setG$ on average, and hence, may constitute a more appropriate parameter setting when $\bar{C} > 1/3$. In practice though, we observed that irrespective of the actual value of $\bar{C}$, it never hurts to increase $\alpha$ to a value  larger than $\max\{1/3,\bar{C}\}$. This effect is illustrated via the following two strategies for generating seed sets for the \textsc{localSearchOQC} algorithm.
\begin{enumerate}
    \item[(\textbf{S1}):] In this strategy, from the NDP of a graph, we select all vertices whose neighborhood density lies in the interval $[0.70,0.95]$. On average, this yields a small number of $20-30$ vertices, with the worst-case extremes arising in the case of the $\textsc{facebook-A}$ graph, where $3.5\%$ of the $4,039$ vertices (a total of $153$) where returned and the $\textsc{Amazon}$ graph, where only $4$ vertices were returned. Every such vertex $v$ is then combined with its neighborhood $\setN(v)$ to generate a seed set $\{v\} \cup \setN(v)$, which is used as initialization for the \textsc{localSearchOQC} algorithm with $\alpha = 1$. For this choice of $\alpha$, the edge-surplus objective function $f_{\alpha}(\setS)$ attaches a high penalty to any subset of vertices which do not form a clique, i.e., we ``encourage" the algorithm to discover cliques.
    \item[(\textbf{S2}):] In an alternative strategy, we partition the interval of neighborhood density values $[0.70,0.95)$ into $5$ sets of disjoint, equi-spaced sub-intervals $[0.7,0.75),[0.75,0.8),[0.8,0.85)$,\\
    $[0.85,0.9)$, and $[0.9,0.95)$. Next, we list the vertices of the graph whose neighborhood edge-densities lie in one of these 5 sub-intervals. For graphs with small $C_g$ the size of the list was always $< 1\%$ of the total number of vertices, whereas it was up to $5\%$ for larger $C_g$. From each sub-interval, we select the vertex whose neighborhood subgraph attains the highest edge-surplus value according to \eqref{eq:OQC}, where the parameter $\alpha$ in $f_{\alpha}(\setS)$ is set to the lower bound of the sub-interval; e.g., for the sub-interval $[0.9-0.95)$, $\alpha = 0.9$. This vertex $v$ is then combined with its neighborhood to form the seed set $\{v\} \cup \setN(v)$, which is then used to initialize \textsc{localSearchOQC}, with the same value of $\alpha$ as the sub-interval lower bound. A total of $5$ such seed sets are generated (one for each sub-interval). In this case, our objective is to induce the algorithm to unearth large quasi-cliques.
\end{enumerate}

\noindent The performance of \textsc{localSearchOQC} using the seeding strategy \textbf{(S1)}, is depicted in Figure \ref{fig:S1} on $2$ representative datasets. By setting $\alpha = 1$, \textsc{localSearchOQC} is indeed capable of discovering cliques when initialized from appropriate vertex neighborhoods. While the size of the discovered cliques is smaller than the largest ego-clique for a small number of seeds, the majority of trials produced cliques of larger sizes. We empirically verified that these cliques are maximal, which concurs with our intuition regarding the algorithm, i.e., if the current solution set is a non-maximal clique, by design, the algorithm will seek to add vertices which will produce a larger, maximal clique (note that the extreme setting $\alpha = 1$ discourages any other vertices from being added in this case). A list of these maximal cliques of size larger than the largest ego-clique for the datasets considered are depicted in the right-hand column of Figure \ref{fig:S1}. We point out that on the \textsc{web-Google} dataset, a few seeds produced subgraphs of small size and low density. This illustrates a potential drawback of setting $\alpha = 1$: if the initial seed set is not in the local vicinity of a denser subgraph, then \textsc{localSearchOQC} compensates by seeking out a small subgraph with low density. To appreciate this behavior, we focus on one such seed set of size $60$ and density $0.77$. For this subgraph, the edge-surplus objective function has a value of $-407$. When used as initialization for \textsc{localSearchOQC}, the algorithm yields a subgraph of size $11$ and edge-density $0.18$. However, the objective function $f_{\alpha}(\setS)$ has a value of $-45$, which marks a near $10$-fold improvement over the initial set. While this is a worst-case scenario for such a  ``all-or-nothing'' approach, we observed that it seldom occurs in practice (only $6/32$ trials on the \textsc{web-Google} graph and no such occurrences on the \textsc{Facebook-B} graph). Overall, our experiments indicate that these vertex neighborhoods can indeed serve as favorable
initialization points for discovering maximal cliques using \textsc{localSearchOQC}.

As a performance benchmark, we also added the \textsc{greedyOQC} algorithm of \cite{tsourakakis2013denser}, with $\alpha$ also set equal to $1$. Interestingly, the algorithm \emph{always} produced a clique with this setting on all the datasets we tried. With regard to detecting cliques, Figure \ref{fig:S1} reveals that the performance of \textsc{greedyOQC} is competitive with \textsc{localSearchOQC}. On the \textsc{Facebook-B} graph, \textsc{localSearchOQC} detects $3$ distinct cliques of size $25$, while \textsc{greedyOQC} also discovers a different clique of the same size. Finally, on the  \textsc{web-Google} graph, the size of the largest clique discovered by \textsc{localSearchOQC} is $43$, which is comparable in size to the largest clique on $46$ vertices produced by \textsc{greedyOQC}. We also empirically observed that the clique returned by \textsc{greedyOQC} does not subsume any of the smaller cliques produced by \textsc{localSearchOQC}, thereby highlighting the contrasting nature of the two approaches. 
\begin{figure}[t!]
	\includegraphics[width = 0.23\textwidth]{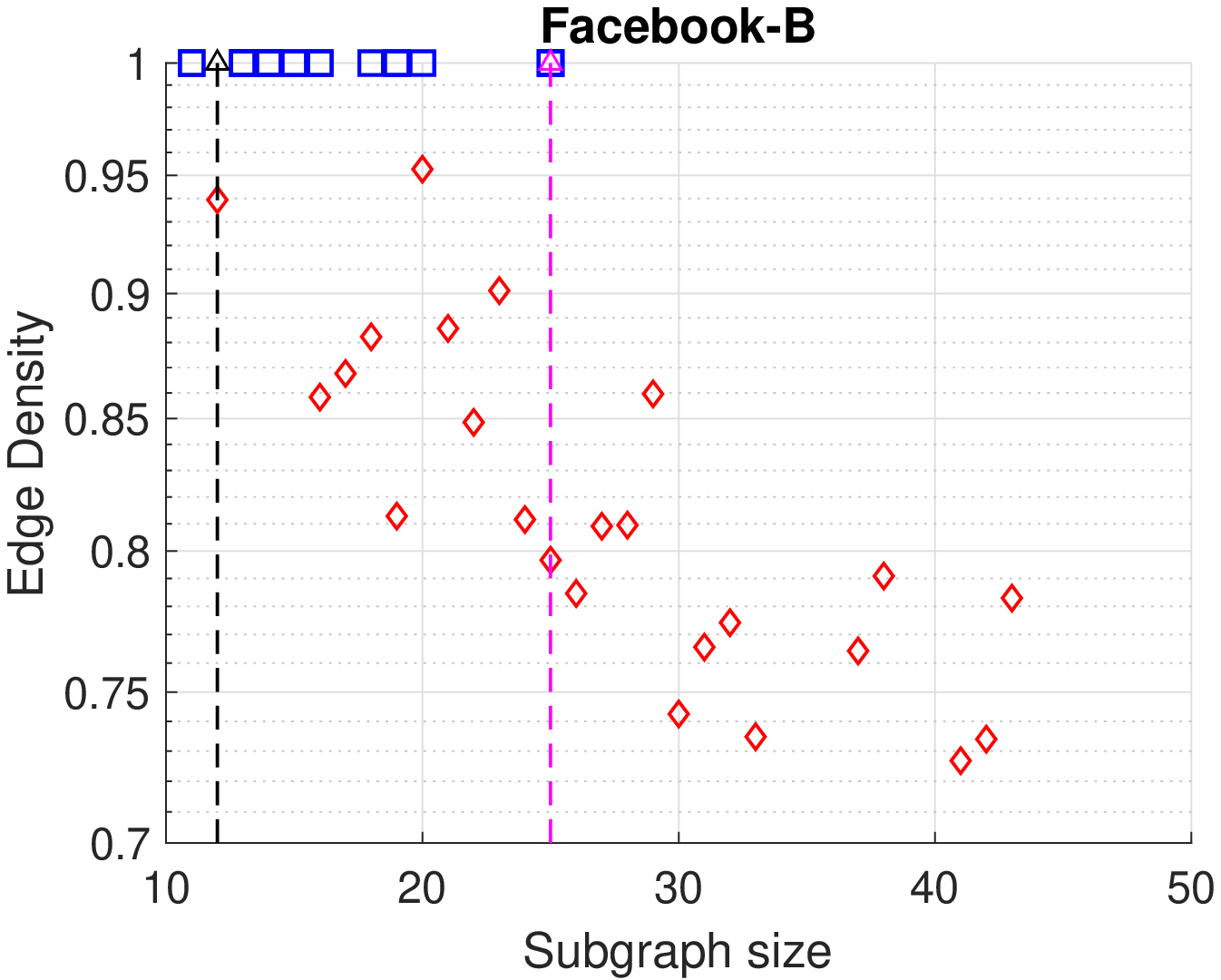}
	\includegraphics[width = 0.23\textwidth]{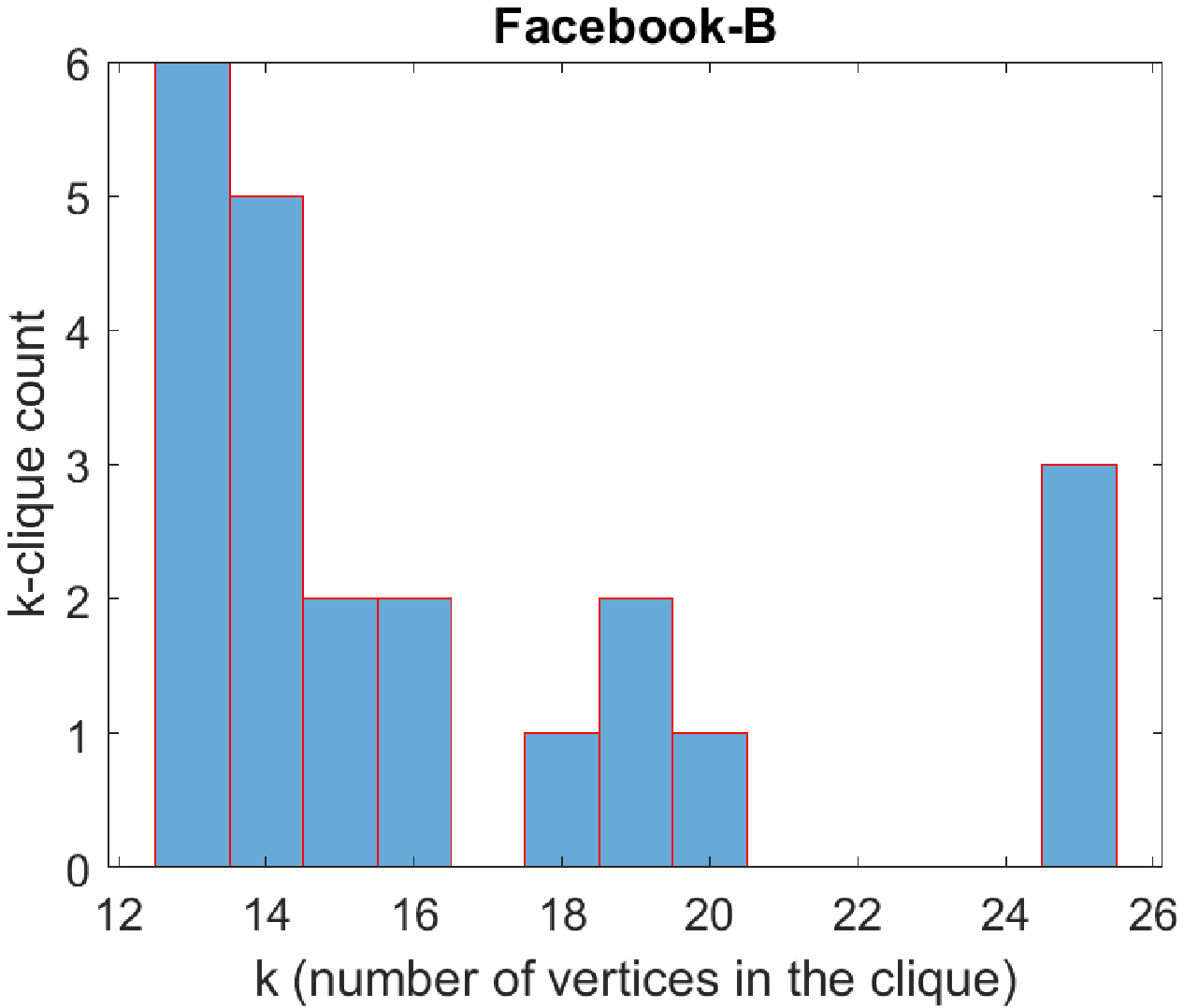}
	\includegraphics[width = 0.23 \textwidth]{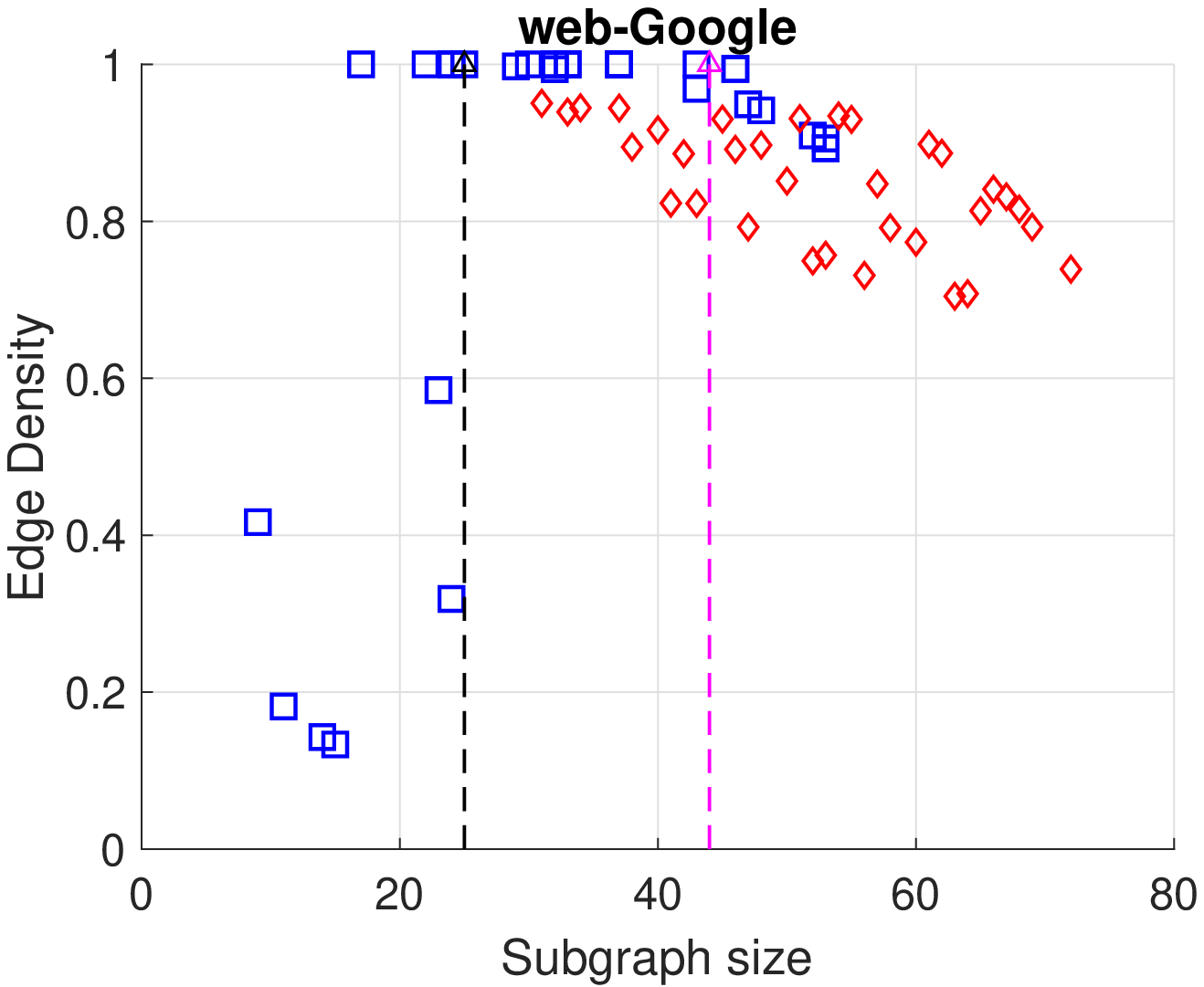}
	\includegraphics[width = 0.23 \textwidth]{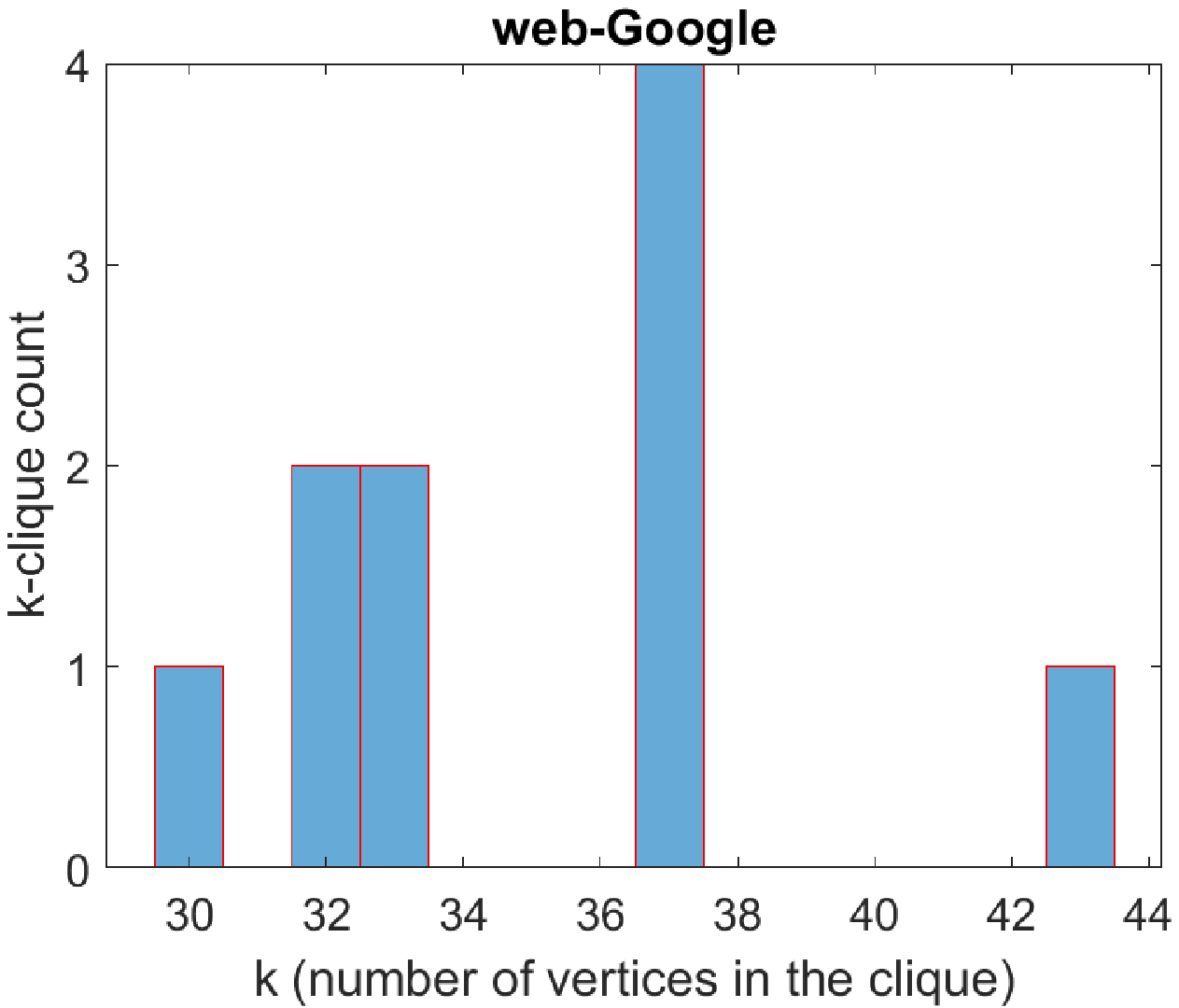}
    \caption{\footnotesize Results of using \textsc{localSearchOQC} with seeds (\textbf{S1}) on three real-world graphs. Left column:
    Edge Density versus subgraph size. The {\color{red} red diamonds} denote the neighborhood subgraphs selected using the seeding strategy (\textbf{S1}), the black vertical line highlights the size of the largest ego-clique, the {\color{blue} blue squares} denote the subgraphs obtained using \textsc{localSearchOQC} with seeds (\textbf{S1}) and $\alpha = 1$, and the {\color{magenta} magenta} vertical line marks the size of the largest clique returned by \textsc{greedyOQC}. Right column: list of $k$-cliques obtained by \textsc{localSearchOQC} of size larger than the largest ego-clique.}  
    \label{fig:S1}
\end{figure}  
\newline \indent We now focus on the effectiveness of \textsc{localSearchOQC} in discovering large quasi-cliques when using the seeding strategy \textbf{(S2)}. We used  \textsc{greedyOQC} again as a benchmark, with the range of parameter settings varying from $\alpha \in \{1/3,0.7,0.75,0.8,0.85,0.9,1\}$, i.e., from the recommended setting $1/3$ to the highest possible value $1$. 
Figure \ref{fig:S2} displays the results of our experiments on $4$ datasets, which are representative of all the possible outcomes that we observed. Regarding the performance of \textsc{greedyOQC}, we point out that the recommended setting of $\alpha = 1/3$ can be \emph{very} sub-optimal with respect to the neighborhood subgraphs we selected. For example, on the \textsc{blogCatalog3} dataset, using $\alpha = 1/3$ outputs a subgraph on $330$ vertices with edge-density $0.5$, which is $33\%$ less dense and $10$ times larger in size than the least-dense neighborhood subgraph obtained. The algorithm demonstrates marked improvement only upon using a more aggressive choice of $\alpha$, with the subgraph size decreasing and the density increasing progressively as $\alpha$ is increased, and ultimately yielding a clique when $\alpha=1$. 
\begin{figure}[t!]
	\includegraphics[width = 0.23\textwidth]{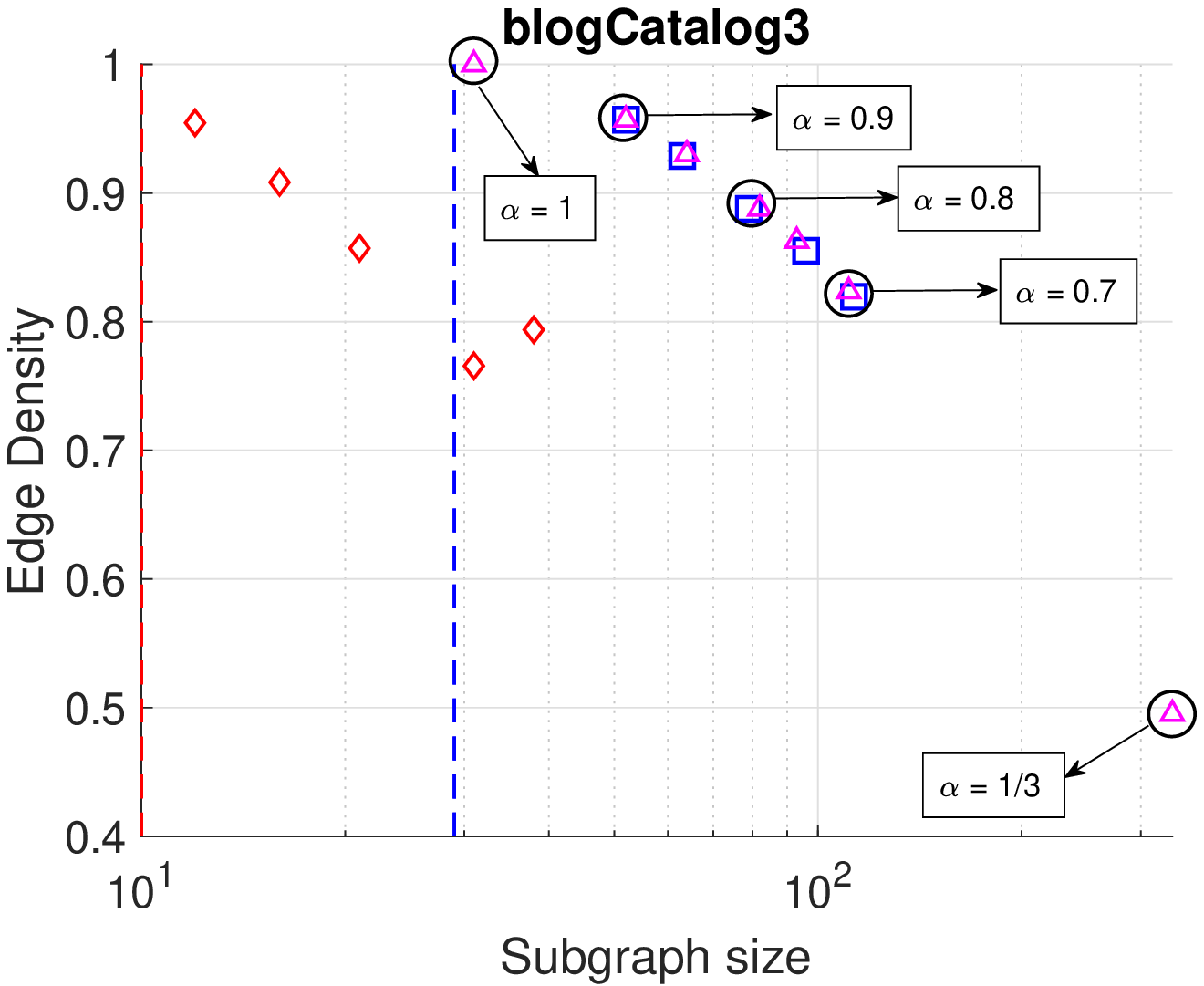}
	\includegraphics[width = 0.23 \textwidth]{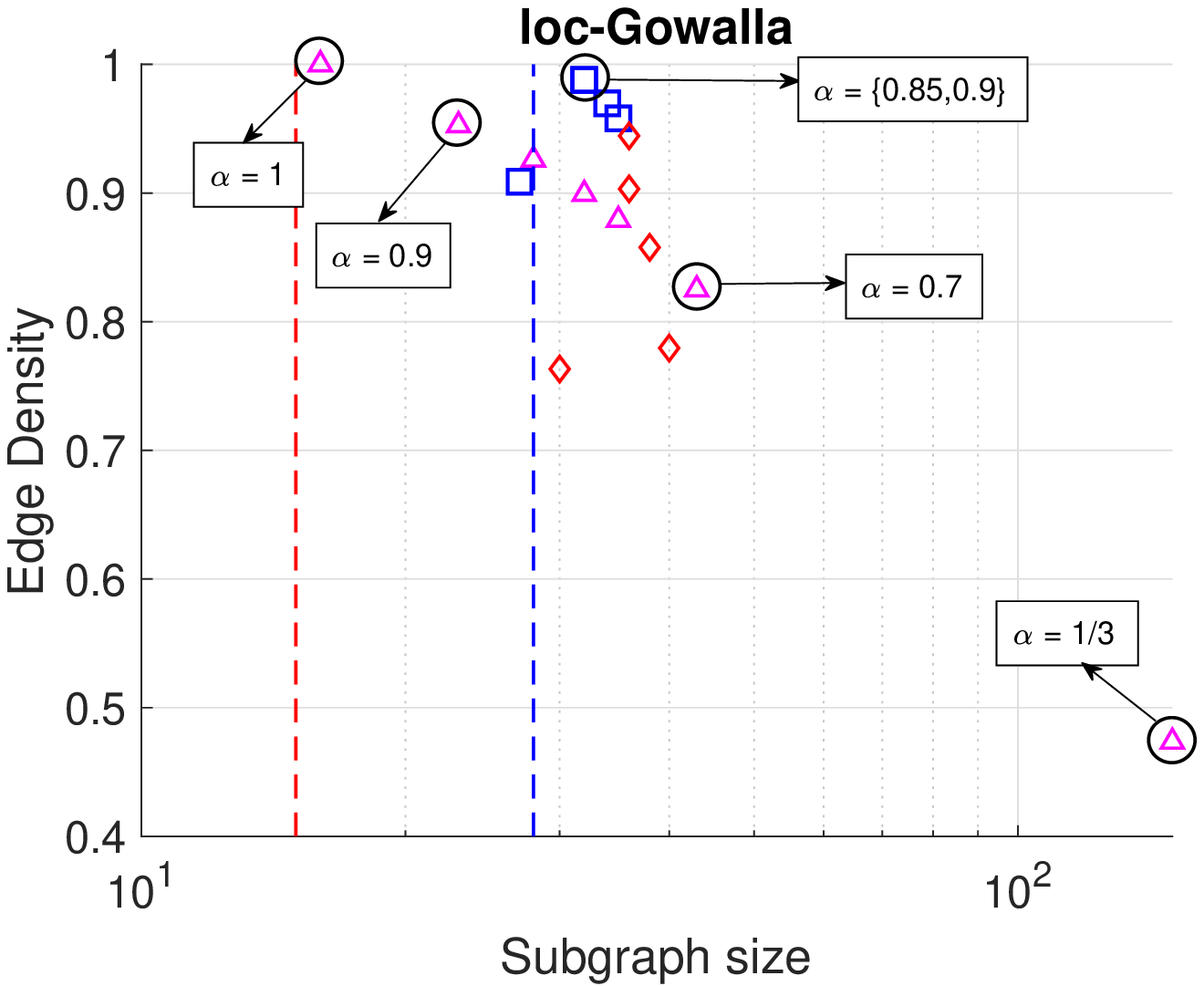}
	\includegraphics[width = 0.23 \textwidth]{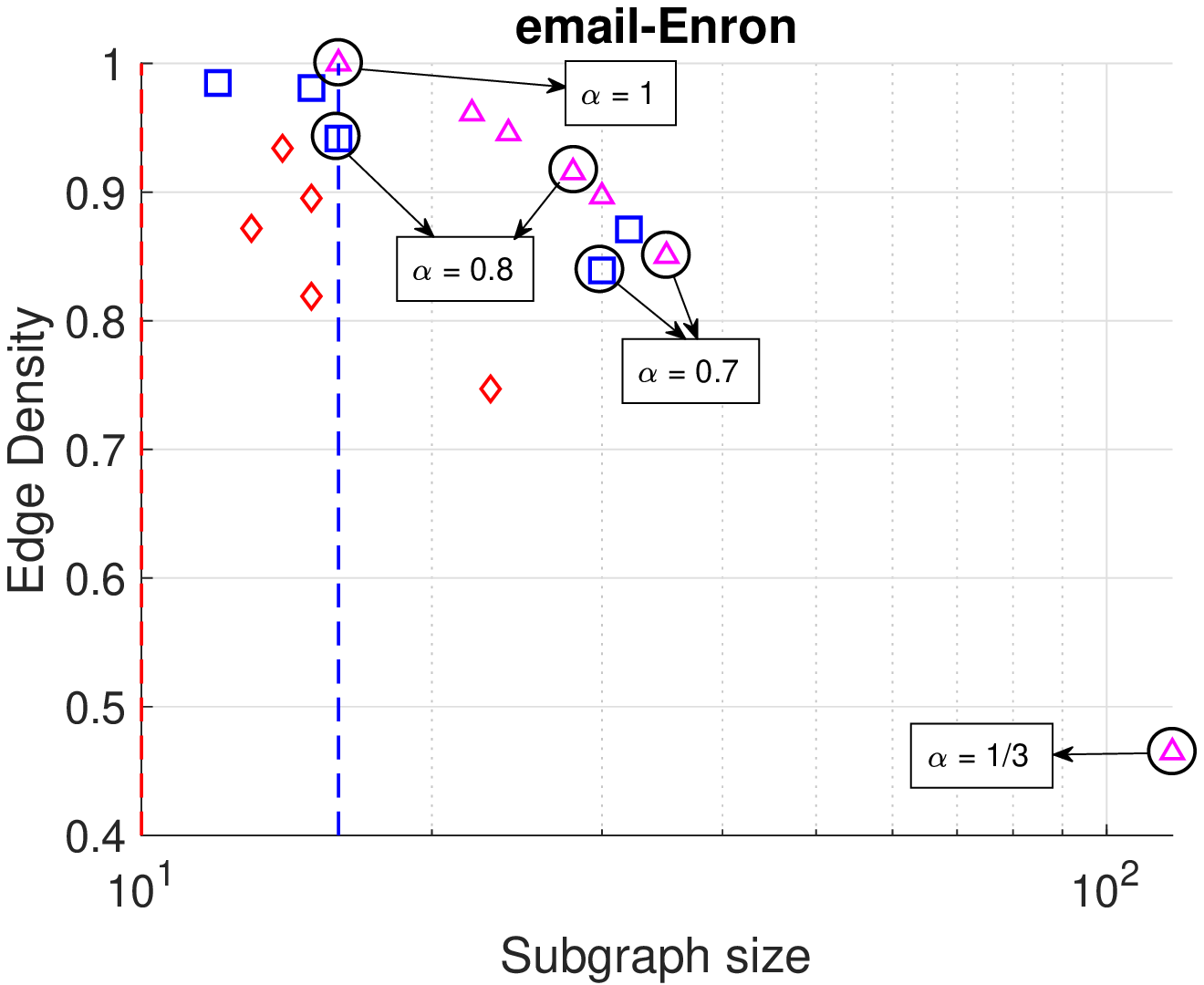}
	\includegraphics[width = 0.23 \textwidth]{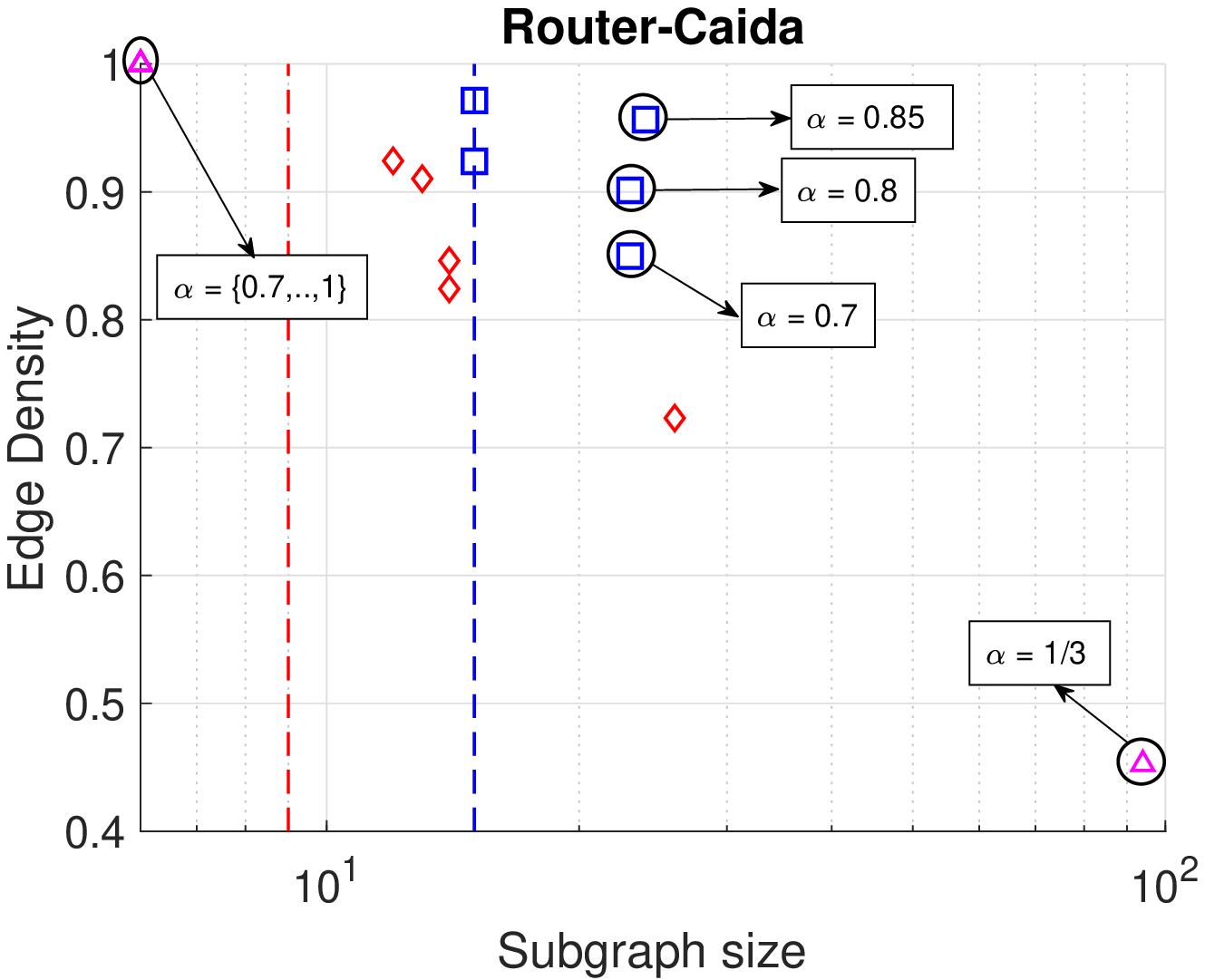}
    \caption{\footnotesize Edge density versus subgraph size for four real-world graphs as a function of the parameter $\alpha$ used in \textsc{localSearchOQC} and \textsc{greedyOQC}. The {\color{red} red diamonds} denote the neighborhood subgraphs selected using the seeding strategy (\textbf{S2}), the {\color{red} red} vertical line highlights the size of the largest ego-clique, the {\color{blue} blue squares} denote the subgraphs obtained using \textsc{localSearchOQC} with seeds (\textbf{S2}), the {\color{blue} blue} vertical line marks the size of the largest clique obtained using \textsc{localSearchOQC} with seeds (\textbf{S1}), and the {\color{magenta} magenta triangles} denote the output of \textsc{greedyOQC}.}  
    \label{fig:S2}
\end{figure}

On the \textsc{blogCatalog3} dataset, in terms of size and edge-density, the quasi-cliques computed by \textsc{localSearchOQC} are a close match to those computed by \textsc{greedyOQC} for a given $\alpha$.  On the other hand, on the \textsc{loc-Gowalla} graph, it can be noted that the initial seed sets themselves are large quasi-cliques. In this case, further refinement using \textsc{localSearchOQC} does not result in a significant improvement, although it does identify a near-clique on $32$ vertices. In comparison, the largest clique detected by \textsc{greedyOQC} is only marginally larger than the largest ego-clique, and is much smaller than the largest clique recovered by \textsc{localSearchOQC}.
On the \textsc{email-Enron} graph, we observe the opposite trend, i.e.,  \textsc{localSearchOQC} produces dense quasi-cliques of smaller size compared to \textsc{greedyOQC} overall. On the
\textsc{router-Caida} graph, we made a curious observation regarding \textsc{greedyOQC} - the subgraph produced is invariant with respect to all choices of $\alpha > 1/3$. In this case, the algorithm completely fails to unveil any dense quasi-cliques, while \textsc{localSearchOQC} discovers a $0.95$-quasi-clique on $24$ vertices. Furthermore, it can be seen that the clique computed by \textsc{greedyOQC} is of size $6$, which is smaller than \emph{both} the largest ego-clique and the largest clique computed by \textsc{localSearchOQC}. 
\begin{figure}[t]
    \centering
    \includegraphics[width = 0.3\textwidth]{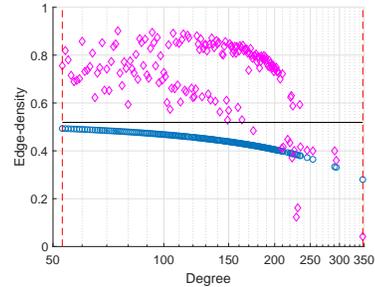}
    \caption{\footnotesize Lower bound of Theorem 3.5 ({\color{blue} blue}) vs actual neighborhood edge-density ({\color{magenta} magenta}) as a function of the degree for the \textsc{Facebook-A} graph. Black line -- $C_g$, {\color{red} red} lines -- admissible range of degrees.}
    \label{fig:bound}
\end{figure}
Finally, we compare the lower bound on the neighborhood edge-density derived in Theorem 3.5 against its actual value for the \textsc{Facebook-A} graph in Figure \ref{fig:bound}. We chose this particular dataset as it has a large value of $C_g = 0.52$, and its degree distribution closely conforms with our assumptions (C1)--(C2). Note that for a fixed $C_g$, the lower bound $(C_g - \beta)/(1-\beta)$ decreases monotonically with $\beta \in (d_{\min}/d_{\max},C_g)$. We plot the value of this lower bound for every unique degree in the graph that lies between a fraction $\beta_{\min} = 0.05$ and $\beta_{\max} = C_g$ of the largest degree $d_{\max} = 1,045$, and also plot the largest clustering coefficient $C_v$ (i.e., the actual neighborhood edge-density) for every such degree. The figure reveals that our lower bound is pessimistic in general, although it becomes tighter for larger degrees. A very small number of neighborhoods of large degree also violate the lower bound, which we attribute to the fact that there are missing degrees in practice and that the degree distribution approximately obeys a power-law with exponent $2$.

\end{document}